\documentclass[journal,12pt,one column]{IEEEtran}
\usepackage{amsfonts}
\usepackage{amssymb}
\usepackage{amsmath}
\usepackage[dvips]{graphicx}

\input{tcilatex}
\newtheorem{defi}{Definition}
\newtheorem{fact}{Fact}
\newtheorem{theo}{Theorem}
\newtheorem{lemm}{Lemma}

\newtheorem{rema}{Remark}

\begin{document}

\title{Eigen-Direction Alignment Based Physical-Layer Network Coding for
MIMO Two-Way Relay Channels}
\author{Tao Yang, \emph{Member, IEEE}, Xiaojun Yuan, \emph{Member, IEEE}, Li
Ping, \emph{Fellow, IEEE}, Iain B. Collings, \emph{Senior Member,
IEEE} and Jinhong Yuan, \emph{Senior Member, IEEE} \thanks{Tao
Yang and Iain. B. Collings are with CSIRO ICT Centre, Australia.
The work of Xiaojun Yuan and Li Ping was fully supported by a
grant from the University Grants Committee of the Hong Kong
Special Administrative Region, China (Project No. AoE/E-02/08).
The work of Jinhong Yuan was supported by Australian Research
Council under the ARC Discovery Grant DP110104995.} } \maketitle

\begin{abstract}
In this paper, we propose a novel communication strategy which
incorporates physical-layer network coding (PNC) into
multiple-input multiple output (MIMO) two-way relay channels
(TWRCs). At the heart of the proposed scheme lies a new key
technique referred to as eigen-direction alignment (EDA)
precoding. The EDA precoding efficiently aligns the two-user's
eigen-modes into the same directions. Based on that, we carry out
multi-stream PNC over the aligned eigen-modes. We derive an
achievable rate of the proposed EDA-PNC scheme, based on nested
lattice codes, over a MIMO TWRC. Asymptotic analysis shows that
the proposed EDA-PNC scheme approaches the capacity upper bound as
the number of user antennas increases towards infinity. For a
finite number of user antennas, we formulate the design criterion
of the optimal EDA precoder and present solutions. Numerical
results show that there is only a marginal gap between the
achievable rate of the proposed EDA-PNC scheme and the capacity
upper bound of the MIMO TWRC, in the median-to-large SNR region.
We also show that the proposed EDA-PNC scheme significantly
outperforms existing amplify-and-forward and decode-and-forward
based schemes for MIMO TWRCs.
\end{abstract}

\newpage

\section{Introduction}

A two-way relay channel (TWRC), where two users exchange
information simultaneously via an intermediate relay, can
potentially double the throughput of a conventional one-way relay
channel \cite{ZhangMobicom06}. Recently, it has been shown that
physical-layer network coding (PNC) can achieve within 1/2 bit of
the capacity of a single-input single-output (SISO) Gaussian TWRC
\cite{NamIT09}, \cite{WilsonIT10}, and it is asymptotically
optimal at high signal-to-noise ratios (SNRs). In the PNC scheme,
the two users transmit signals simultaneously to the relay. The
relay recovers and forwards only compressed information of the two
users, rather than the complete information. This is in contrast
to the well-known amplify-and-forward (AF)
\cite{KattiSigcomm07}-\cite{XuICCASP2010} and decode-and-forward
(DF) based schemes \cite{GunduzAsilomar2008} for TWRCs.

The existing work on PNC is limited to SISO scenarios. It is
well-known that multiple-input multiple-output (MIMO) systems can
provide many advantages over SISO systems, in a rich-scattering
environment \cite{Foschini96}. The
challenge is to extend PNC to MIMO TWRCs. In \cite{ZhangMobicom06} and \cite%
{NamIT09}, the PNC scheme required that the two-user's signals
received by the relay are \textit{aligned}\ in the same spatial
direction. This
condition is naturally guaranteed in a SISO Gaussian TWRC \cite%
{ZhangMobicom06}, \cite{NamIT09}. However, in a MIMO environment,
each user has multiple eigen-modes. The directions of the
eigen-modes (referred to as eigen-directions) of the two users in
the TWRC are different in general. Therefore, the main challenge
is to design an efficient technique to align the eigen-directions
of the two users. This will lead to a practical PNC scheme for
MIMO TWRC. We will show that the performance can be up to 50\%
higher in spectral efficiency at practical SNR levels, compared
with the existing schemes for MIMO TWRCs that do not employ PNC.

In this paper, we propose a novel \textit{eigen-direction
alignment} (EDA) precoding based PNC scheme for MIMO TWRCs. \ The
key of the proposed EDA precoding is that it efficiently aligns
the two-user's eigen-modes into the same directions. Then, we
construct multiple independent PNC streams over the aligned
eigen-modes established by the EDA precoding. We refer to the
proposed strategy as an EDA-PNC scheme.

We derive achievable rates of the proposed EDA-PNC scheme, based
on nested lattice codes \cite{NamIT09}. Our asymptotic analysis
shows that the proposed EDA-PNC scheme approaches the capacity
upper bound of a MIMO TWRC, as the numbers of user antennas
increase towards infinity. For a finite number of user antennas,
we formulate the design criterion of the optimal EDA precoder,
which leads to a non-convex optimization problem. For a relatively
small spatial dimension, we develop an exhaustive search method to
obtain the optimal EDA precoder. For a larger spatial dimension,
we derive \textit{approximate solutions} to the optimization
problem. Numerical results show that there is only a marginal gap
between the achievable rate of the proposed EDA-PNC scheme and the
capacity upper bound of the MIMO TWRC, in the median-to-large SNR
region. We also show that the proposed EDA-PNC scheme
significantly outperforms the existing AF- and DF-based schemes
for MIMO TWRC.

The paper is organized as follows. In Section II, we depict the
model of a MIMO TWRC and a two-phase transmission protocol. In
Section III, we derive a capacity upper bound and briefly discuss
two existing schemes. In Section IV, we propose the EDA-PNC
scheme. In Section V, we derive the achievable rate of the
proposed scheme and present an asymptotical result. The design
criterion of the optimal EDA precoder is also given in Section V.
In Section VI, we discuss sub-optimal EDA precoders. Numerical
results are shown and discussed in Section VII. Finally, we draw
the conclusions in Section VIII.

\section{System Model}

In this section, we introduce the modelling of a MIMO TWRC and
describe a two-phase transmission protocol. We focus on a
real-valued model in this paper. The extension of our results to a
complex-valued model is straightforward, as detailed in Appendix
I.

\subsection{Configuration of a MIMO TWRC}

A MIMO TWRC, in which user $A$ and user $B$ exchange information
via a relay, is illustrated in Fig.
\ref{Fig_Configuration_MIMOTWRC}. Each user is equipped with
$n_{T}$ antennas and the relay has $n_{R}$ antennas. All the
channels in the system are assumed to be flat-fading within the
bandwidth of interest. The channel from user $A$ $($or $B)$ to the
relay is denoted by an
$n_{R}$-by-$n_{T}$ matrix $\mathbf{H}_{A,R}$ $\left( \text{or }\mathbf{H}%
_{B,R}\right) $. The channel from\ the relay to user $A$ $($or
$B)$ is
denoted by an $n_{T}$-by-$n_{R}$ matrix $\mathbf{H}_{R,A}$ $\left( \text{or }%
\mathbf{H}_{R,B}\right) $.

The users and the relay operate in half-duplex mode. There is no
direct link between the two users. The transmission protocol
employs two consecutive equal-duration time-slots for each round
of information exchange between the users via the relay. Each
time-slot consists of $n$ channel users. In the first time-slot
(uplink phase), the two users transmit to the relay simultaneously
and the relay remains silent. In the second time-slot (downlink
phase), the relay broadcasts to the two silent users. We assume
that the channel coefficients remain the same for each round of
information exchange. We also assume that the channel matrices are
globally known by both users, as well as by the relay.

In this paper, we will only consider the situation of $n_{T}$ $\geq $ $n_{R}$%
. This configuration applies to practical scenarios such as a
wireless sensor network where the physical sizes of the
intermediate sensor nodes are
smaller than those of the terminal nodes\footnote{%
The situation of $n_{T}<n_{R}$ will be addressed in our future
work.}.

\subsection{Uplink Phase}

The discrete channel of the uplink phase can be written as%
\begin{equation}
Y_{R}\left[ l\right] =\mathbf{H}_{A,R}X_{A}\left[ l\right] +\mathbf{H}%
_{B,R}X_{B}\left[ l\right] +Z_{R}\left[ l\right] ,\text{
}l=1,\cdots ,n, \label{Eq_SystemModel_Uplink}
\end{equation}%
where $X_{m}\left[ l\right] $ is an $n_{T}$-by-1 column vector with the $i$%
th entry $x_{m,i}\left[ l\right] ,i=1,\cdots ,n_{T}$, being the
coded signal transmitted from antenna $i$ of user $m$, $m\in
\left\{ A,B\right\} $, at time instant $l$; $Y_{R}\left[ l\right]
$ is an $n_{R}$-by-1 column vector with the $j$th entry
$y_{R,j}\left[ l\right] ,j=1,\cdots ,n_{R}$, being the signal
received from antenna $j$ of the relay; $Z_{R}\left[ l\right] $ is
an $n_{R}$-by-1 additive white Gaussian noise (AWGN) vector at the
relay with the $j$th entry $z_{R,j}\left[ l\right] \sim
\mathcal{N}\left( 0,\sigma _{R}^{2}\right) ,$ $j=1,\cdots ,n_{R}$,
where $\sigma _{R}^{2}$ is the noise variance. For notational
simplicity, the time index $l$ may be omitted in situations
without causing ambiguity.

The channel input covariances of the two users are denoted by $\mathbf{Q}%
_{m}=\mathcal{E}\left( X_{m}X_{m}^{T}\right) ,m\in \left\{
A,B\right\} $, where $\mathcal{E}\left( \cdot \right) $ stands for
the expectation operation. The power constraint of the uplink
phase is given by
\begin{equation}
\text{Tr}\left\{ \mathbf{Q}_{A}+\mathbf{Q}_{B}\right\} \leq P_{T}
\label{Eq_PowerConstraint}
\end{equation}%
where $P_{T}$ is the total transmission power of the two users.
The average
per-user SNR of the uplink phase is defined as%
\begin{equation}
SNR\triangleq \frac{P_{T}}{2\sigma _{R}^{2}}.
\label{Eq_DefinitionSNR}
\end{equation}

\subsection{Relay's Operation}

Upon receiving $\mathbf{Y}_{R}=\left[ Y_{R}\left[ 1\right] ,\cdots ,Y_{R}%
\left[ n\right] \right] $, the relay generates a signal matrix $\mathbf{X}%
_{R}=\left[ X_{R}\left[ 1\right] ,\cdots ,X_{R}\left[ n\right]
\right] $.
Here, $X_{R}\left[ l\right] $ is an $n_{R}$-by-$1$ real vector with the $j$%
th entry $x_{R,j}\left[ l\right] $, $j=1,\cdots ,n_{R}$, being the
signal transmitted from the $j$th antenna of the relay, at time
instant $l$, in the
downlink phase. In general, the relationship between $\mathbf{X}_{R}$ and $%
\mathbf{Y}_{R}$ can be written as%
\begin{equation}
\mathbf{X}_{R}=f_{R}\left( \mathbf{Y}_{R}\right)
\label{Eq_RelayFunctionality}
\end{equation}%
where $f_{R}\left( \cdot \right) $ denotes the relay's
functionality. The
relay's power constraint is given by%
\begin{equation}
\text{Tr}\left\{ \mathbf{Q}_{R}\right\} \leq P_{R}
\label{Eq_RelayPowerConstraint}
\end{equation}%
where $\mathbf{Q}_{R}=\mathcal{E}\left( X_{R}X_{R}^{T}\right) \
$denotes the channel input covariance matrix of the relay in the
downlink phase.

\begin{rema}
In this paper, the power constraints under consideration are given by $%
\left( \ref{Eq_PowerConstraint}\right) $ and $\left( \ref%
{Eq_RelayPowerConstraint}\right) $. The generalization to the case
with a global sum-power constraint can be readily done by trading
off the portion of power allocated to the users and that to the
relay.
\end{rema}

\subsection{Downlink Phase}

During the downlink phase, the signal $\mathbf{X}_{R}=\left[ X_{R}\left[ 1%
\right] ,\cdots ,X_{R}\left[ n\right] \right] $ serves as the
channel input
and is broadcast to users $A$ and $B$. The signals received by user $m$, $%
m\in \left\{ A,B\right\} ,$ are given by%
\begin{equation}
Y_{m}\left[ l\right] =\mathbf{H}_{R,m}X_{R}\left[ l\right] +Z_{m}\left[ l%
\right] ,l=1,\cdots ,n,  \label{Eq_SystemModel_Downlink}
\end{equation}%
where $Z_{m}\left[ l\right] $ is an $n_{T}$-by-1 AWGN vector with
the $i$th entry $z_{m,i}\left[ l\right] \sim \mathcal{N}\left(
0,\sigma _{m}^{2}\right) $, $i=1,\cdots ,n_{T}$, where $\sigma
_{m}^{2}$ is the noise
variance at user $m$. Upon receiving $\mathbf{Y}_{A}=\left[ Y_{A}\left[ 1%
\right] ,\cdots ,Y_{A}\left[ n\right] \right] $, user $A$ decodes
user $B$'s message with the help of the perfect knowledge of
$\mathbf{X}_{A}=\left[ X_{A}\left[ 1\right] ,\cdots ,X_{A}\left[
n\right] \right] $. Meanwhile, similar operations are performed by
user $B$. This finishes one round of information exchange.

For notational simplicity, we assume $\sigma _{R}^{2}=\sigma
_{A}^{2}=\sigma _{B}^{2}=1$ in this paper. Then, the average
per-user SNR is $SNR=P_{T}/2$ and the SNR of the relay is
$SNR_{R}=$ $P_{R}$. The extension of our results to the case of
unequal noise power is straightforward.

\section{Capacity Upper Bound and Existing Schemes for a MIMO TWRC}

\subsection{Definitions}

The achievable rate-pair and rate-region of a MIMO TWRC are
defined as follows:

\begin{defi}
\label{Definition1}A rate-pair $\left( R_{A}\text{, }R_{B}\right)
$ is said
to be achievable if there exists a set of $2^{nR_{A}}$ codewords for user $A$%
, a set of $2^{nR_{B}}$ codewords for user $B$ and a relay functionality $%
\mathbf{X}_{R}=f_{R}\left( \mathbf{Y}_{R}\right) ,$ satisfying
power constraints (\ref{Eq_PowerConstraint}) and
(\ref{Eq_RelayPowerConstraint}), such that the decoding error
probabilities approach zero at both user nodes of the TWRC, as
$n\rightarrow \infty $.
\end{defi}

\begin{rema}
The rate of each user is defined as the amount of transmitted bits
in each transmission round, normalized by the duration of
\textit{one} phase (consisting of $n$ channel uses).
\end{rema}

\begin{defi}
The achievable rate-region $\mathcal{R}$ is defined as the convex
closure of all achievable rate-pairs$.$
\end{defi}

\subsection{Capacity Upper Bound of a MIMO TWRC}

We now derive a new capacity upper bound (UB) for a MIMO TWRC. We
present the result in the following lemma, which is an extension
of the cut-set bound for a SISO TWRC \cite{NamIT09}.

\begin{lemm}
\label{LemmaAchievableRateUB}For given input covariance matrices $\mathbf{Q}%
_{A},$ $\mathbf{Q}_{B}$ and $\mathbf{Q}_{R}$, the achievable
rate-pair of a MIMO TWRC is upper bounded by
\begin{subequations}
 \label{Eq_Rate_UB_MIMO}
\begin{align}
& R_{A}\leq R_{A}^{UB}=\frac{1}{2}\min \left[ \log \det \left( \mathbf{I}+%
\mathbf{H}_{A,R}\mathbf{Q}_{A}\mathbf{H}_{A,R}^{T}\right) \text{,
}\log \det
\left( \mathbf{I}+\mathbf{H}_{R,B}\mathbf{Q}_{R}\mathbf{H}_{R,B}^{T}\right) %
\right]  \label{Eq_RateA_UB_MIMO} \\
& R_{B}\leq R_{B}^{UB}=\frac{1}{2}\min \left[ \log \det \left( \mathbf{I}+%
\mathbf{H}_{B,R}\mathbf{Q}_{B}\mathbf{H}_{B,R}^{T}\right) \text{,
}\log \det
\left( \mathbf{I}+\mathbf{H}_{R,A}\mathbf{Q}_{R}\mathbf{H}_{R,A}^{T}\right) %
\right] .  \label{Eq_RateB_UB_MIMO}
\end{align}
\end{subequations}
\end{lemm}

\begin{proof}
From the cut-set bound, the achievable rate-pair of a TWRC is
upper-bounded
by \cite{NamIT09}%
\begin{subequations}
\begin{align}
& R_{A}\leq R_{A}^{UB}=\min \left\{
I(X_{A};Y_{R}|X_{B}),I(X_{R};Y_{B})\right\} ,  \label{Eq_RateA_Simplified} \\
& R_{B}\leq R_{B}^{UB}=\min \left\{
I(X_{B};Y_{R}|X_{A}),I(X_{R};Y_{A})\right\} .
\label{Eq_RateB_Simplified}
\end{align}%
\end{subequations}
Applying the capacity formula of a real-valued MIMO channel \cite%
{TSETextBook} for given input covariances $\mathbf{Q}_{A},$
$\mathbf{Q}_{B}$ and $\mathbf{Q}_{R}$, we obtain
(\ref{Eq_Rate_UB_MIMO}).
\end{proof}

With the result in Lemma \ref{LemmaAchievableRateUB}, the capacity
UB of a
MIMO TWRC can be determined by optimizing\footnote{%
This is a convex optimization problem which can be easily solved
using a
standard tool, e.g., \cite{CVXSoftware}.} the covariance matrices $\mathbf{Q}%
_{A},$ $\mathbf{Q}_{B}$ and $\mathbf{Q}_{R}$. This capacity UB
provides an upper limit on the data rate that any MIMO two-way
relay scheme can achieve.

\subsection{Analog Network Coding for a MIMO TWRC}

Much progress has been made in developing communication strategies
to approach the capacity of a MIMO TWRC. Among those, an AF-based
scheme,
namely analog network coding (ANC) \cite{KattiSigcomm07}-\cite{XuICCASP2010}%
, has attracted a great deal of attention. In ANC, the relay
broadcasts an amplified version of its received signal to the two
users. The maximum achievable rate of ANC in MIMO TWRC remains
unsolved and a sub-optimal solution is reported in
\cite{XuICCASP2010}. In this paper, we will consider the upper
bound of the achievable rate of ANC, derived in \cite{YuanANC}, as
a benchmark for comparison purpose.

The ANC scheme has two disadvantages. First, it suffers from noise
amplification, since the noise received by the relay is not
suppressed before the signal is forwarded to the users. Second, it
suffers from unnecessary power consumption at the relay, since the
AF relay forwards a linear, rather than an algebraic,
superposition of the signals from the two users
\cite{ZhangMobicom06}.

\subsection{DF with Network Coding for a MIMO TWRC}

A DF-based scheme has also been studied for a MIMO TWRC \cite%
{GunduzAsilomar2008}, \cite{ZhangJSAC09}. In the DF-based scheme,
the relay completely decodes both users' messages. The decoded
messages of two users
are re-encoded with a network code \cite{ZhangMobicom06}, \cite%
{AhlswedeIT2000}, and a channel code. The resultant coded signal
is broadcast to the two users in the downlink phase. We refer to
this scheme as DF with network coding (DF-NC).

The achievable rate of the DF-NC scheme is briefly discussed as
follows. The uplink phase of DF-NC can be viewed as a MIMO
multiple-access channel whose exact achievable rate-region is
still an open problem, although its upper and lower bounds are
studied in \cite{YuIT04}. We will use the upper bound in
\cite{YuIT04} for comparison purpose. The downlink rate-region of
the DF-NC scheme can be obtained by extending the result in
\cite{NamIT09} to a MIMO scenario. The overall achievable
rate-region of the DF-NC scheme is the intersection of the uplink
and downlink rate-regions determined above.

The DF-NC scheme suffers from a severe multiplexing loss
\cite{WilsonIT10}, \cite{GunduzAsilomar2008}, as complete decoding
at the relay is demanding and unnecessary. As a result, the
achievable rate of the DF-NC scheme may far below the capacity of
a MIMO TWRC, especially in the high SNR region
\cite{GunduzAsilomar2008}.

\section{Eigen-Direction Alignment Based Physical-Layer Network Coding}

In this section, we propose a new strategy for MIMO TWRCs. The
proposed strategy consists of two key components:
\textit{eigen-direction alignment (EDA) precoding\ }and\textit{\
physical-layer network coding} (PNC). In particular, the proposed
EDA precoding algorithm efficiently aligns the eigen-directions of
two users. Then, we carry out multi-stream PNC over the aligned
eigen-modes established by the EDA precoding.

To illustrate the proposed EDA precoding algorithm, we first
describe a straightforward (naive) method to perform the
eigen-direction alignment.

\subsection{A Naive Eigen-Direction Alignment Approach}

Denote by $\mathbf{F}_{A}$ and $\mathbf{F}_{B}$ the linear
precoding matrices of user $A$ and user $B$, respectively. The
users' transmitted
signals can be written as%
\begin{equation}
X_{m}\left[ l\right] =\mathbf{F}_{m}C_{m}\left[ l\right] ,m\in
\left\{ A,B\right\} ,l=1,\cdots ,n,  \label{Eq Transmitted signal
vector}
\end{equation}%
where $C_{m}\left[ l\right] =\left[ c_{m,1}\left[ l\right] ,\cdots
,c_{m,n_{R}}\left[ l\right] \right] ^{T}$, $\mathcal{E}\left(
C_{m}C_{m}^{T}\right) =\mathbf{I}$, is a length-$n_{R}$ column
vector whose entries denote the independently coded signals. As a
straightforward approach, the precoder performs channel inverse,
i.e., $\mathbf{F}_{m}$ is
given by\footnote{%
In general, a randomly generated $\mathbf{H}_{A,R}$ $\left( \text{or }%
\mathbf{H}_{B,R}\right) $ with $n_{T}\geq n_{R}$ is of full
row-rank with probability 1 \cite{TSETextBook}. For simplicity of
discussion, we always
assume that $\mathbf{H}_{A,R}$ and $\mathbf{H}_{B,R}$ are of full row-rank.}%
\begin{equation}
\mathbf{F}_{m}=\mathbf{H}_{m,R}^{T}\left( \mathbf{H}_{m,R}\mathbf{H}%
_{m,R}^{T}\right) ^{-1}\mathbf{\Psi }_{m},m\in \left\{ A,B\right\}
, \label{Eq ZF}
\end{equation}%
where $\mathbf{H}_{m,R}^{T}\left( \mathbf{H}_{m,R}\mathbf{H}%
_{m,R}^{T}\right) ^{-1}$ is the Moore-Penrose pseudo-inverse of $\mathbf{H}%
_{m,R}$ (for $n_{T}\geq n_{R}$) and $\mathbf{\Psi }_{m}$ is an
$n_{R}\times n_{R}$ diagonal matrix which allocates power among
the $n_{R}$ eigen-modes
for user $m$. With (\ref{Eq ZF}), the signal received by the relay in (\ref%
{Eq_SystemModel_Uplink}) can be written as%
\begin{subequations}
\begin{align}
Y_{R}\left[ l\right] & =\mathbf{H}_{A,R}\mathbf{F}_{A}C_{A}\left[ l\right] +%
\mathbf{H}_{B,R}\mathbf{F}_{B}C_{B}\left[ l\right] +Z_{R}\left[
l\right]
\label{Eq ZF Received Signal} \\
& =\mathbf{\Psi }_{A}C_{A}\left[ l\right] +\mathbf{\Psi }_{B}C_{B}\left[ l%
\right] +Z_{R}\left[ l\right] .  \label{Eq ZF Received Signal b}
\end{align}%
\end{subequations} Eq. (\ref{Eq ZF Received Signal b}) represents
$n_{R}$ parallel sub-channels, as both $\mathbf{\Psi }_{A}$ and
$\mathbf{\Psi }_{B}$ are diagonal matrices. The above approach is
referred to as a \textit{naive EDA precoding}. Unfortunately, it
is well-known that the channel inverse in precoding suffers from a
significant power loss when the channel matrix is ill-conditioned
\cite{TSETextBook}. Thus, this approach may not be an efficient
method to align the eigen-directions.

\subsection{Proposed Eigen-Direction Alignment Precoding}

Now, we propose our new EDA precoding algorithm which can
effectively avoid the power loss suffered by the naive EDA
precoder. Consider an invertible
linear transformation of the relay's received signal as%
\begin{eqnarray}
\widetilde{Y}_{R}\left[ l\right] &=&\mathbf{K}^{-1}Y_{R}\left[
l\right]
\label{Eq Linear Transformation} \\
&=&\mathbf{K}^{-1}\mathbf{H}_{A,R}X_{A}\left[ l\right] +\mathbf{K}^{-1}%
\mathbf{H}_{B,R}X_{B}\left[ l\right] +\mathbf{K}^{-1}Z_{R}\left[ l\right] ,%
\text{ }l=1,\cdots ,n,  \notag
\end{eqnarray}%
where $\mathbf{K}$ is an $n_{R}$-by-$n_{R}$ invertible square
matrix referred to as the \textit{rotation matrix}. The equivalent
channel matrices
are now given by%
\begin{equation}
\widetilde{\mathbf{H}}_{m,R}=\mathbf{K}^{-1}\mathbf{H}_{m,R},m\in
\left\{ A,B\right\} .  \label{Eq Equivalent Channel}
\end{equation}

Applying the aforementioned naive EDA precoding over the
equivalent channel in (\ref{Eq Equivalent Channel}), we obtain the
proposed new EDA precoding
matrix as%
\begin{eqnarray}
\mathbf{F}_{m} &=&\widetilde{\mathbf{H}}_{m,R}^{T}\left( \widetilde{\mathbf{H%
}}_{m,R}\widetilde{\mathbf{H}}_{m,R}^{T}\right) ^{-1}\mathbf{\Psi
}_{m}
\notag \\
&=&\mathbf{H}_{m,R}^{T}\left(
\mathbf{H}_{m,R}\mathbf{H}_{m,R}^{T}\right) ^{-1}\mathbf{K\Psi
}_{m},m\in \left\{ A,B\right\} .  \label{Eq_Precoder}
\end{eqnarray}%
The signal received by the relay in (\ref{Eq_SystemModel_Uplink})
can then
be written as%
\begin{equation}
Y_{R}\left[ l\right] =\mathbf{H}_{A,R}\mathbf{F}_{A}C_{A}\left[ l\right] +%
\mathbf{H}_{B,R}\mathbf{F}_{B}C_{B}\left[ l\right] +Z_{R}\left[ l\right] =%
\mathbf{K}\left( \mathbf{\Psi }_{A}C_{A}\left[ l\right] +\mathbf{\Psi }%
_{B}C_{B}\left[ l\right] \right) +Z_{R}\left[ l\right]
\end{equation}%
where $l=1,\cdots ,n$. At the relay, after the linear transformation (\ref%
{Eq Linear Transformation}), we obtain
\begin{equation}
\widetilde{Y}_{R}\left[ l\right] =\mathbf{\Psi }_{A}C_{A}\left[ l\right] +%
\mathbf{\Psi }_{B}C_{B}\left[ l\right] +\widetilde{Z}_{R}\left[
l\right] \label{Eq_Relay_Signal_Multiply_K Step2}
\end{equation}%
where $\widetilde{Z}_{R}\left[ l\right]
=\mathbf{K}^{-1}Z_{R}\left[ l\right] $ is the equivalent noise
vector. From (\ref{Eq_Relay_Signal_Multiply_K Step2}), it is clear
that $n_{R}$ aligned eigen-modes are established. Note that we can
always scale the entries of $\widetilde{Y}_{R}$ such that the
equivalent noises of all eigen-modes have unit power. Thus,
without loss of generality, we confine the rotation matrix
$\mathbf{K}$ that the diagonal elements of $\mathbf{K}^{-1}\left(
\mathbf{K}^{-1}\right) ^{T}$ are 1, i.e.,
\begin{equation}
\left[ \mathbf{K}^{-1}\left( \mathbf{K}^{-1}\right) ^{T}\right] _{\text{diag}%
}=\mathbf{I}\text{.}  \label{Eq K Constraint}
\end{equation}%
This is to ensure that the entries in the effective noise vector $\widetilde{%
Z}_{R}\left[ l\right] $ have unit power.

The proposed EDA precoding scheme reduces to the naive EDA scheme
by letting $\mathbf{K}$ $=\mathbf{I}$. By varying the rotation
matrix $\mathbf{K}$, we can actually align the eigen-modes of the
two users into any $n_{R}$ pre-determined directions in the
$n_{R}$-dimension vector space, as illustrated in Fig.
\ref{Fig_ZF_ECA}. An immediate question is how to determined the
optimal rotation matrix $\mathbf{K}$. We will retain the answer to
this problem till the next section.

\subsection{The Overall Proposed EDA-PNC Scheme}

We now describe a multi-stream PNC scheme. In the uplink phase,
the proposed EDA precoding (\ref{Eq_Precoder}) is employed to
establish $n_{R}$ aligned parallel sub-channels. The two users
perform single-stream PNC for each aligned sub-channel, and there
are $n_{R}$ independent PNC streams in total.
Similarly to the case of SISO PNC \cite{NamIT09}, the relay recovers the $%
n_{R}$ bin-indices (as defined in \cite{NamIT09})\ instead of
completely decoding both users' individual messages. In the
downlink phase, the aggregation of the $n_{R}$ bin-indices is
re-encoded and broadcast to the two users. Finally, each user
recovers the other user's message with the help of the perfect
knowledge of its own message.

\section{Achievable Rates of the Proposed EDA-PNC Scheme for MIMO TWRCs}

\subsection{Achievable Rate-Pair}

We now present a theorem on the achievable rate-pair of the
proposed EDA-PNC scheme. Define $x^{+}\triangleq \max \left(
x,0\right) .$

\begin{theo}
\label{Theorem2}For given $\mathbf{K}$, $\mathbf{\Psi }_{A}$, $\mathbf{\Psi }%
_{B}$ and $\mathbf{Q}_{R}$, an achievable rate-pair of the
proposed EDA-PNC scheme is given by
\begin{subequations}
\label{EqTheorem2}
\begin{align}
R_{A}& \leq \min \left\{ R_{A,UL}^{EDA},R_{A,DL}^{EDA}\right\}
\label{EqATheorem2} \\
R_{B}& \leq \min \left\{ R_{B,UL}^{EDA},R_{B,DL}^{EDA}\right\}
\label{EqBTheorem2}
\end{align}%
\end{subequations}
where
\begin{subequations}
\begin{align}
R_{A,UL}^{EDA}& =\frac{1}{2}\dsum\limits_{i=1}^{n_{R}}\left[ \log
\left( \frac{\Psi _{A}\left( i,i\right) ^{2}}{\Psi _{A}\left(
i,i\right) ^{2}+\Psi _{B}\left( i,i\right) ^{2}}+\Psi _{A}\left(
i,i\right) ^{2}\right) \right]
^{+},  \label{Eq RateA UL original} \\
R_{B,UL}^{EDA}& =\frac{1}{2}\dsum\limits_{i=1}^{n_{R}}\left[ \log
\left( \frac{\Psi _{B}\left( i,i\right) ^{2}}{\Psi _{A}\left(
i,i\right) ^{2}+\Psi _{B}\left( i,i\right) ^{2}}+\Psi _{B}\left(
i,i\right) ^{2}\right) \right]
^{+},  \label{Eq RateB UL original} \\
R_{A,DL}^{EDA}& =\frac{1}{2}\log \det \left( \mathbf{I}+\mathbf{H}_{R,B}%
\mathbf{Q}_{R}\mathbf{H}_{R,B}^{T}\right) ,  \label{Eq RateA DL} \\
R_{B,DL}^{EDA}& =\frac{1}{2}\log \det \left( \mathbf{I}+\mathbf{H}_{R,A}%
\mathbf{Q}_{R}\mathbf{H}_{R,A}^{T}\right) .  \label{Eq RateB DL}
\end{align}
\end{subequations}
\end{theo}

The proof of Theorem \ref{Theorem2} is given in Appendix II. The
main idea
of the proof is to utilize the results on nested lattice codes in \cite%
{NamIT09}.

\subsection{An Asymptotic Result on the Achievable Rate-Pair}

We next derive an asymptotic result which is based on the
following observation.

\begin{fact}
Assume that the entries of the channel matrix $\mathbf{H}_{m,R}$
are i.i.d.
with zero mean and unit variance. Then,%
\begin{equation}
\frac{1}{n_{T}}\mathbf{H}_{m,R}\mathbf{H}_{m,R}^{T}\overset{\text{P}}{%
\rightarrow }\mathbf{I}\text{, as }n_{T}\rightarrow \mathbf{\infty
,} \label{Eq Convergence to one}
\end{equation}%
where \textquotedblleft $\overset{\text{P}}{\rightarrow }$%
\textquotedblright\ represents convergence in probability.
\end{fact}

The above result is straightforward by invoking the weak law of
large numbers.

\begin{theo}
\label{Theorem4}Assume that the channel coefficients in
$\mathbf{H}_{A,R}$
and $\mathbf{H}_{B,R}$ are i.i.d. withe zero mean and unit variance. As $%
n_{T}$ tends to infinity (while $n_{R}$ remains finite), the
proposed EDA-PNC scheme asymptotically achieves the capacity of a
MIMO TWRC in probability. $\square $
\end{theo}

The proof of Theorem \ref{Theorem4} is given in Appendix III. Theorem \ref%
{Theorem4} states that the proposed EDA-PNC scheme achieves the
capacity
upper bound of a MIMO TWRC with probability 1 as $n_{T}\rightarrow \mathbf{%
\infty }$. This asymptotic result will be verified by the
numerical results presented later.

\subsection{Determining the Achievable Rate-Region}

Here, we consider the achievable rate-region of the proposed
EDA-PNC scheme, based on the results of Theorem \ref{Theorem2}.
Define the following rate-regions
\begin{subequations}
\begin{eqnarray}
\mathcal{R}_{UL}^{EDA} &\triangleq &\left\{ \left(
R_{A},R_{B}\right) :R_{A}\leq R_{A,UL}^{EDA},R_{B}\leq
R_{B,UL}^{EDA}\right\} ,
\label{Eq Rate Region UL} \\
\mathcal{R}_{DL}^{EDA} &\triangleq &\left\{ \left(
R_{A},R_{B}\right) :R_{A}\leq R_{A,DL}^{EDA},R_{B}\leq
R_{B,DL}^{EDA}\right\} .
\end{eqnarray}%
\end{subequations}
The above two rate-regions will be respectively determined in the
following.

\subsubsection{Uplink Rate-Region}

The boundary of the uplink rate-region $\mathcal{R}_{UL}^{EDA}$
can be determined by solving the following weighted sum-rate (WSR)
problem
\begin{subequations}
\label{Eq WSR_Constraint Optimal EDA Precoder}
\begin{equation}
\max_{\mathbf{K,\Psi }_{A},\text{ }\mathbf{\Psi }_{B}}\left\{
\alpha R_{A,UL}^{EDA}+\left( 1-\alpha \right)
R_{B,UL}^{EDA}\right\} \label{Eq WSR Optimal EDA Precoder}
\end{equation}%
subject to%
\begin{equation}
\text{Tr}\left( \left( \mathbf{H}_{A,R}\mathbf{H}_{A,R}^{T}\right) ^{-1}%
\mathbf{K\Psi }_{A}^{2}\mathbf{K}^{T}+\left( \mathbf{H}_{B,R}\mathbf{H}%
_{B,R}^{T}\right) ^{-1}\mathbf{K\Psi
}_{B}^{2}\mathbf{K}^{T}\right) \leq P_{T}.  \label{Eq Constraint
Optimal EDA Precoder}
\end{equation}%
\end{subequations} and (\ref{Eq K Constraint}), for $0\leq \alpha \leq 1$. Note
that the power constraint in (\ref{Eq
Constraint Optimal EDA Precoder}) is obtained by substituting (\ref%
{Eq_Precoder}) and
$\mathbf{Q}_{m}=\mathbf{F}_{m}\mathbf{F}_{m}^{T}$, $m\in \left\{
A,B\right\} $, into (\ref{Eq_PowerConstraint}).

The problem in (\ref{Eq WSR_Constraint Optimal EDA Precoder}) is
non-convex and hence is difficult to solve. For a small $n_{R}$,
e.g., $n_{R}=2$, the optimal\ parameters $\left( \mathbf{K,\Psi
}_{A},\mathbf{\Psi }_{B}\right) $ can be found by an exhaustive
search. Unfortunately, this method quickly becomes prohibitively
complex as $n_{R}$ increases. We will provide approximate
solutions to this problem in Section VI.

\subsubsection{Downlink Rate-Region}

The boundary of the downlink rate-region $\mathcal{R}_{DL}^{EDA}$
can be
determined by solving%
\begin{equation}
\max_{\mathbf{Q}_{R}:\text{ Tr}\left( \mathbf{Q}_{R}\right) \leq
P_{R}}\left\{ \alpha R_{A,DL}^{EDA}+\left( 1-\alpha \right)
R_{B,DL}^{EDA}\right\}  \label{Eq Rate star}
\end{equation}%
for $0\leq \alpha \leq 1$. Note that $\log \det \left( \cdot
\right) \ $is concave and thus the objective function in (\ref{Eq
Rate star}) is concave in $\mathbf{Q}_{R}$. In addition, Tr$\left(
\mathbf{Q}_{R}\right) \leq P_{R}$ is a linear constraint.
Therefore, (\ref{Eq Rate star}) can be solved using convex
optimization.

\subsubsection{Overall Rate-Region}

The overall achievable rate-region of the proposed EDA-PNC scheme
is the intersection of $\mathcal{R}_{UL}^{EDA}$ and
$\mathcal{R}_{DL}^{EDA}$.

The major difficulty in determining the above achievable
rate-region of the proposed EDA-PNC scheme is to solve the WSR
problem in (\ref{Eq WSR_Constraint Optimal EDA Precoder}). In the
next section, we will provide two suboptimal solutions to this
problem.

\section{Approximate Solutions to the Optimal EDA Precoder}

\subsection{Approximate Solution I}

To simplify the problem in (\ref{Eq WSR_Constraint Optimal EDA
Precoder}), we introduce
two extra constraints on the proposed EDA precoder: 1) The rotation matrix $%
\mathbf{K}$ is unitary, i.e.,
\begin{equation}
\mathbf{KK}^{T}=\mathbf{I},  \label{Limitation 1}
\end{equation}%
and 2) The power matrices satisfy
\begin{equation}
\mathbf{\Psi }_{A}=\mathbf{\Sigma },\text{ }\mathbf{\Psi
}_{B}=\gamma \mathbf{\Sigma },  \label{Limitation 2}
\end{equation}%
where $\gamma $ is a positive scalar and $\mathbf{\Sigma }$ is a
diagonal matrix with non-negative diagonal elements. Although
these extra constraints may lead to a certain performance loss, a
close-form solution then exists, which yields crucial insights
into the design of the EDA precoder. Later, we will consider the
relaxation of these two constraints to obtain a better approximate
solution.

\subsubsection{Optimal Unitary Rotation Matrix $\mathbf{K}$ (for $\mathbf{%
\Psi }_{B}=\protect\gamma \mathbf{\Psi }_{A}$)}

Now we derive the most power-efficient unitary rotation matrix
$\mathbf{K}$
for given $\mathbf{\Psi }_{A}=\mathbf{\Sigma }$ and $\mathbf{\Psi }%
_{B}=\gamma \mathbf{\Sigma }$. The problem is formulated as%
\begin{equation}
\mathbf{K}_{opt}^{(\mathbf{\Sigma },\gamma )}=\arg \min_{\mathbf{K:}\text{ }%
\mathbf{KK}^{T}=\mathbf{I}}\text{Tr}\left( \mathbf{F}_{A}\mathbf{F}_{A}^{H}+%
\mathbf{F}_{B}\mathbf{F}_{B}^{H}\right)  \label{Eq Problem}
\end{equation}

Let the singular value decomposition (SVD) of the channel matrix $\mathbf{H}%
_{m,R}$ be
\begin{equation}
\mathbf{H}_{m,R}=\mathbf{U}_{m}\mathbf{\Sigma
}_{m}\mathbf{V}_{m}^{T},\text{ }m\in \left\{ A,B\right\} ,
\label{Eq SVD}
\end{equation}%
where $\mathbf{U}_{m}$ and $\mathbf{V}_{m}^{T}$ are unitary matrices and $%
\mathbf{\Sigma }_{m}$ is an $n_{R}$-by-$n_{T}$ diagonal matrix
with positive diagonal elements. Denote by $\mathbf{\Sigma
}_{m}^{-1}$ the pseudo-inverse of $\mathbf{\Sigma }_{m}$, i.e.,
\begin{equation*}
\mathbf{\Sigma }_{m}^{-1}=\left[
\begin{array}{c}
\mathbf{\Upsilon }_{m}^{-1} \\
\mathbf{0}_{(n_{T}-n_{R})\times n_{R}}%
\end{array}%
\right]
\end{equation*}%
where $\mathbf{\Upsilon }_{m}$ is an $n_{R}$-by-$n_{R}$ matrix
formed by the
first $n_{R}$ columns of $\mathbf{\Sigma }_{m}$ and $\mathbf{0}%
_{(n_{T}-n_{R})\times n_{R}}$ denotes an
$(n_{T}-n_{R})$-by-$n_{R}$ matrix
with all-zero entries. For notational simplicity, we denote $\left( \mathbf{%
\Sigma }_{m}\mathbf{\Sigma }_{m}^{T}\right) ^{-1}$ by $\mathbf{\Sigma }%
_{m}^{-2}$. Then, using (\ref{Eq_Precoder}) and (\ref{Eq SVD}), the problem (%
\ref{Eq Problem}) becomes
\begin{equation}
\mathbf{K}_{opt}^{(\mathbf{\Sigma },\gamma )}=\arg \min_{\mathbf{K:}\text{ }%
\mathbf{KK}^{T}=\mathbf{I}}\text{Tr}\left( \left( \mathbf{U}_{A}\mathbf{%
\Sigma }_{A}^{-2}\mathbf{U}_{A}^{T}+\gamma ^{2}\mathbf{U}_{B}\mathbf{\Sigma }%
_{B}^{-2}\mathbf{U}_{B}^{T}\right) \mathbf{K\Sigma K}^{T}\right) .
\label{Eq_K_opt}
\end{equation}%
Define
\begin{equation}
\mathbf{G}\left( \gamma \right) \triangleq \mathbf{U}_{A}\mathbf{\Sigma }%
_{A}^{-2}\mathbf{U}_{A}^{T}+\gamma ^{2}\mathbf{U}_{B}\mathbf{\Sigma }%
_{B}^{-2}\mathbf{U}_{B}^{T}.  \label{Eq Definition of G}
\end{equation}%
The eigen-decomposition of $\mathbf{G}\left( \gamma \right) $ yields $%
\mathbf{G}\left( \gamma \right) =\mathbf{U}_{G\left( \gamma \right) }\mathbf{%
\Lambda }_{G\left( \gamma \right) }\mathbf{U}_{G\left( \gamma
\right) }^{T}$ where $\mathbf{\Lambda }_{G\left( \gamma \right) }$
is a diagonal matrix
with the diagonal entries arranged in the \textit{ascending order}, and $%
\mathbf{U}_{G\left( \gamma \right) }$ is a unitary matrix. Without
loss of generality, we always assume that the diagonal entries of
$\mathbf{\Sigma }$ are arranged in the \textit{descending order}.
Now, we present the optimal unitary rotation matrix $\mathbf{K}$
in the following theorem.

\begin{theo}
\label{Theorem1}For any given $\mathbf{\Psi }_{A}=\mathbf{\Sigma }$ and $%
\mathbf{\Psi }_{B}=\gamma \mathbf{\Sigma }$, the solution to the problem in (%
\ref{Eq_K_opt}) is%
\begin{equation}
\mathbf{K}_{opt}^{(\gamma )}=\mathbf{U}_{G\left( \gamma \right) }.
\label{Eq Theorem 1}
\end{equation}
\end{theo}

\begin{proof}
With (\ref{Eq Definition of G}), the objective function in
(\ref{Eq_K_opt})
is written as%
\begin{equation}
\text{Tr}\left( \mathbf{G}\left( \gamma \right) \mathbf{K\Sigma }^{2}\mathbf{%
K}^{T}\right) =\text{Tr}\left( \mathbf{U}_{G\left( \gamma \right) }\mathbf{%
\Lambda }_{G\left( \gamma \right) }\mathbf{U}_{G\left( \gamma \right) }^{T}%
\mathbf{K\Sigma }^{2}\mathbf{K}^{T}\right) \leq \text{Tr}\left( \mathbf{%
\Lambda }_{G\left( \gamma \right) }\mathbf{\Sigma }^{2}\right) .
\label{Eq Step1 Proof of Theorem1}
\end{equation}%
where the equality in the last step holds when
$\mathbf{K=U}_{G\left( \gamma \right) }$. The inequality in
(\ref{Eq Step1 Proof of Theorem1}) follows the
fact \cite{TangTWC07}, \cite{HornTextbook}: for any two hermitian matrix $\mathbf{M}$ and $%
\mathbf{N}$ with eigen decomposition $\mathbf{M}=\mathbf{U}_{M}\mathbf{%
\Lambda }_{M}\mathbf{U}_{M}^{T}$ and $\mathbf{N}=\mathbf{U}_{N}\mathbf{%
\Lambda }_{N}\mathbf{U}_{N}^{T}$,
\begin{equation}
\text{Tr}\left( \mathbf{MN}\right) \leq \text{Tr}\left( \mathbf{\Lambda }_{M}%
\mathbf{\Lambda }_{N}\right)  \label{Eq_Trace_inequlity}
\end{equation}%
where the diagonal elements of $\mathbf{\Lambda }_{M}$ and those $\mathbf{%
\Lambda }_{N}$ are reversely ordered. This finishes the proof.
\end{proof}

We have the following comments on Theorem \ref{Theorem1}.

\begin{rema}
\label{Corollary 1}The optimal unitary rotation matrix $\mathbf{K}%
_{opt}^{(\gamma )}$ is dependent of $\gamma $, but not of $\mathbf{\Sigma }$%
. Thus, we write $\mathbf{K}_{opt}^{(\gamma )}$ instead of $\mathbf{K}%
_{opt}^{(\mathbf{\Sigma },\gamma )}$.
\end{rema}

\begin{rema}
With $\mathbf{K}_{opt}^{(\gamma )}=\mathbf{U}_{G\left( \gamma
\right) }$, the power constraint in (\ref{Eq_PowerConstraint}) can
be expressed as
\begin{equation}
\text{Tr}\left[ \mathbf{G}\left( \gamma \right) \mathbf{K}_{opt}^{(\gamma )}%
\mathbf{\Sigma }^{2}\left( \mathbf{K}_{opt}^{(\gamma )}\right) ^{T}\right] =%
\text{Tr}\left( \mathbf{\Lambda }_{G\left( \gamma \right) }\mathbf{\Sigma }%
^{2}\right) \leq P_{T}\text{.}  \label{Eq Theorem1 Power
Constraint}
\end{equation}
\end{rema}

\subsubsection{Uplink Rate-Region Revisited}

Here, we present an approximate solution to the the uplink
achievable
rate-region $\mathcal{R}_{UL}^{EDA}$ of the proposed EDA-PNC scheme. With $%
\mathbf{\Psi }_{A}=\mathbf{\Sigma }$ and $\mathbf{\Psi }_{B}=\gamma \mathbf{%
\Sigma }$, (\ref{Eq RateA UL original}) and (\ref{Eq RateB UL
original}) become
\begin{subequations}
\label{Eq Rate UL}
\begin{align}
\widetilde{R}_{A,UL}^{EDA}&
=\frac{1}{2}\dsum\limits_{i=1}^{n_{R}}\left[
\log \left( \frac{1}{1+\gamma ^{2}}+\Sigma \left( i,i\right) ^{2}\right) %
\right] ^{+},  \label{Eq RateA UL} \\
\widetilde{R}_{B,UL}^{EDA}&
=\frac{1}{2}\dsum\limits_{i=1}^{n_{R}}\left[ \log \left(
\frac{\gamma ^{2}}{1+\gamma ^{2}}+\gamma ^{2}\Sigma \left(
i,i\right) ^{2}\right) \right] ^{+}  \label{Eq RateB UL}
\end{align}%
\end{subequations}
where $\Sigma \left( i,i\right) $ denotes the $i$th diagonal entry of $%
\mathbf{\Sigma }$.

Correspondingly, the WSR problem in (\ref{Eq WSR_Constraint
Optimal EDA Precoder}) becomes
\begin{subequations}
\label{Eq Optimization ULRegion_Both}
\begin{equation}
\max_{\mathbf{\Sigma },\text{ }\gamma ,\mathbf{K}}\left\{ \alpha \widetilde{R%
}_{A,UL}^{EDA}+\left( 1-\alpha \right)
\widetilde{R}_{B,UL}^{EDA}\right\} \label{Eq Optimization
ULRegion}
\end{equation}%
subject to
\begin{equation}
\text{Tr}(\mathbf{G}\left( \gamma \right) \mathbf{K\Sigma }^{2}\mathbf{K}%
^{T})\leq P_{T}\text{ and }\mathbf{KK}^{T}=\mathbf{I.} \label{Eq
Optimization ULRegion Constraint}
\end{equation}%
\end{subequations}
For the above problem, if the optimal $\left( \mathbf{K},\mathbf{\Sigma }%
\right) $ couple for any given $\gamma $ can be found, the optimal
solution to (\ref{Eq Optimization ULRegion_Both}) can be easily
determined by a one-dimension full search over $\gamma $.

We next determine the optimal $\left( \mathbf{K},\mathbf{\Sigma
}\right) $ couple for an arbitrarily given $\gamma $. The optimal
unitary rotation matrix $\mathbf{K}$ to the problem in (\ref{Eq
Optimization ULRegion_Both}), for a given $\gamma $, is presented
in the following lemma (which is a direct result of Theorem
\ref{Theorem1}).
\begin{lemm}
Given $\gamma $, the optimal $\mathbf{K}$ to the maximum WSR problem in (\ref%
{Eq Optimization ULRegion_Both}) is $\mathbf{K}_{opt}^{(\gamma )}=\mathbf{U}%
_{G\left( \gamma \right) }$ given in (\ref{Eq Theorem 1}).
\end{lemm}

The remaining task is to find the optimal diagonal matrix $\mathbf{\Sigma }$%
. The optimization problem in (\ref{Eq Optimization ULRegion}) can
be equivalently written as
\begin{equation}
\max_{\mathbf{\Sigma }}\left\{ \frac{\alpha }{2}\dsum\limits_{i=1}^{n_{R}}%
\left[ \log \left( \frac{1}{1+\gamma ^{2}}+\Sigma \left(
i,i\right) ^{2}\right) \right] ^{+}+\frac{1-\alpha
}{2}\dsum\limits_{i=1}^{n_{R}}\left[ \log \left( \frac{\gamma
^{2}}{1+\gamma ^{2}}+\gamma ^{2}\Sigma \left( i,i\right)
^{2}\right) \right] ^{+}\right\} \label{Eq Optimization ULRegion
Plus}
\end{equation}%
subject to Tr$\left( \mathbf{\Lambda }_{G\left( \gamma \right) }\mathbf{%
\Sigma }^{2}\right) \leq P_{T}$ (cf., (\ref{Eq Theorem1 Power
Constraint})).

The objective function (\ref{Eq Optimization ULRegion Plus}) involves $%
[\cdot ]^{+}$ operations, and thus is not concave. However, if we
know in advance which $[\cdot ]^{+}$ operations should be
activated, (\ref{Eq Optimization ULRegion Plus}) can be converted
into a convex optimization problem. Consider any two index subsets
$S_{m}\subseteq $ $\{1,\ldots
,n_{R}\}$, $m\in \left\{ A,B\right\} $. We formulate the following problem%
\begin{equation}
\max_{\mathbf{\Sigma }}\left\{ \frac{\alpha }{2}\dsum\limits_{i\in S_{A}}%
\left[ \log \left( \frac{1}{1+\gamma ^{2}}+\Sigma \left(
i,i\right) ^{2}\right) \right] +\frac{1-\alpha
}{2}\dsum\limits_{i\in S_{B}}\left[ \log \left( \frac{\gamma
^{2}}{1+\gamma ^{2}}+\gamma ^{2}\Sigma \left( i,i\right)
^{2}\right) \right] \right\}   \label{Eq Optimization ULRegion
PlusPlus}
\end{equation}%
subject to Tr$\left( \mathbf{\Lambda }_{G\left( \gamma \right) }\mathbf{%
\Sigma }^{2}\right) \leq P_{T}$. The solution to the above problem
is given in the following lemma, with the proof given in Appendix
IV.

\begin{lemm}
\label{Lemma3}For given $S_{A}$, $S_{B}$ and $\gamma $, the
solution to the problem in (\ref{Eq Optimization ULRegion
PlusPlus}) is given by
\end{lemm}

\begin{equation}
\Sigma _{opt}^{\left( S_{A},S_{B},\gamma \right) }(i,i)=\left\{
\begin{array}{cc}
\sqrt{\left( \frac{1}{2\lambda \Lambda _{G\left( \gamma \right) }(i,i)}-%
\frac{1}{1+\gamma ^{2}}\right) ^{+}} & \text{if }i\in S_{A}\text{
and }i\in
S_{B} \\
\sqrt{\left( \frac{\alpha }{2\lambda \Lambda _{G\left( \gamma \right) }(i,i)}%
-\frac{1}{1+\gamma ^{2}}\right) ^{+}} & \text{if }i\in S_{A}\text{ and }%
i\notin S_{B} \\
\sqrt{\left( \frac{1-\alpha }{2\lambda \Lambda _{G\left( \gamma
\right) }(i,i)}-\frac{1}{1+\gamma ^{2}}\right) ^{+}} & \text{if
}i\notin S_{A}\text{
and }i\in S_{B} \\
0 & \text{if }i\notin S_{A}\text{ and }i\notin S_{B}%
\end{array}%
\right.   \label{Eq Lemma3}
\end{equation}%
where $\lambda $ is a real scalar satisfying%
\begin{equation}
\dsum\limits_{i=1}^{n_{R}}\Lambda _{G\left( \gamma \right)
}(i,i)\left( \Sigma _{opt}^{\left( S_{A},S_{B},\gamma \right)
}(i,i)\right) ^{2}=P_{T}. \label{Eq Lemma3 Constraint}
\end{equation}

Lemma \ref{Lemma3} yields the optimal power matrix $\mathbf{\Sigma }%
_{opt}^{\left( S_{A},S_{B},\gamma \right) }$ for given $S_{A}$ and
$S_{B}$. The optimal power matrix $\mathbf{\Sigma }_{opt}^{\left(
\gamma \right) }$ can be found by evaluating $\mathbf{\Sigma
}_{opt}^{\left( S_{A},S_{B},\gamma \right) }$ for all possible
$\{S_{A},S_{B}\}$.

We now conclude the solution to (\ref{Eq Optimization
ULRegion_Both}) in the following theorem.

\begin{theo}
\label{Theorem3}For any given $\gamma $ , the optimal ($\mathbf{K}$, $%
\mathbf{\Sigma })$ to the problem in (\ref{Eq Optimization
ULRegion_Both})
is given by%
\begin{equation*}
\mathbf{K=K}_{opt}^{(\gamma )}\text{ and }\mathbf{\Sigma }=\mathbf{\Sigma }%
_{opt}^{\left( \gamma \right) }
\end{equation*}%
where $\mathbf{\Sigma }_{opt}^{\left( \gamma \right) }$ is the optimal $%
\mathbf{\Sigma }_{opt}^{\left( S_{A},S_{B},\gamma \right) }$ over
all
possible \{$S_{A},S_{B}$\}, and $\mathbf{K}_{opt}^{(\gamma )}$ is given in (%
\ref{Eq Theorem 1}).
\end{theo}

\begin{proof}
This follows directly from Lemma 2 and Lemma \ref{Lemma3}$.$
\end{proof}

To solve (\ref{Eq Optimization ULRegion Plus}) more efficiently,
we may
confine $S_{A}=S_{B}$. Then, the solution is given by%
\begin{equation}
\widetilde{\Sigma }_{opt}^{\left( \gamma \right) }(i,i)=\sqrt{\left( \frac{1%
}{2\lambda \Lambda _{G\left( \gamma \right) }(i,i)}-\frac{1}{1+\gamma ^{2}}%
\right) ^{+}}.  \label{Eq Lemma3 Revised}
\end{equation}%
It is observed from numerical results that the extra constraint of $%
S_{A}=S_{B}$ incurs un-noticeable performance loss. Finally, we
perform a one-dimension full search over $\gamma$ which yields the
approximate solution. This algorithm is summarized as the
approximate solution I below.

\texttt{\ \ \ \ \ \ Approximate Solution I} \texttt{%
\begin{align*}
& \text{for }\gamma =0\text{ to }1 \ \text{and} \ 1/\gamma
=1\text{ to }0 \text{, with a step }\delta \ \ \ \
\ \ \ \ \ \ \ \ \ \ \ \ \ \ \ \ \ \ \  \\
& \ \ \ \ \ \text{compute }\mathbf{\Lambda }_{G\left( \gamma \right) }\
\text{using (\ref{Eq Definition of G})} \\
& \ \ \ \ \ \ \ \ \ \ \text{{compute}}{\ }\widetilde{{\mathbf{\Sigma }}}{%
_{opt}^{\left( \gamma \right) }\ }\text{{using (\ref{Eq Lemma3 Revised})}} \\
& \ \ \ \ \ \ \ \ \ \ \text{{compute}}{\ \widetilde{R}_{A,UL}^{EDA}\ }\text{{%
and}}{\ \widetilde{R}_{B,UL}^{EDA}\ }\text{{\ in(\ref{Eq RateA UL}) and (\ref%
{Eq RateB UL})}} \\
& \ \ \ \ \ \ \ \ \ \ \text{{backup the corresponding WSR}\newline
} \\
& \text{end} \\
& \text{find the highest WSR in the backup}
\end{align*}%
}

\subsection{Approximate Solution II}

The approximate solution I (AS-I) relies on two constraints: 1) $\mathbf{KK}%
^{T}=\mathbf{I}$ and 2) $\mathbf{\Psi }_{A}=\mathbf{\Sigma }$,
$\mathbf{\Psi }_{B}=\gamma \mathbf{\Sigma }$. We next relax these
constraints to improve the AS-I.

We start with the first constraint. Recall the optimization problem in (\ref%
{Eq_K_opt}). We may ask what is the optimal rotation matrix
$\mathbf{K}$
while relaxing the unitary matrix constraint. This problem is formulated as%
\footnote{%
The constraint (\ref{Eq 2}) implies that, for any given
$\mathbf{\Sigma }$, the achievable rate pair of the EDA-PNC scheme
is fixed.}
\begin{subequations}
\label{Eq 1_2}
\begin{equation}
\min_{\mathbf{K}}\text{Tr}(\mathbf{G}\left( \gamma \right) \mathbf{K\Sigma }%
^{2}\mathbf{K}^{T})  \label{Eq 1}
\end{equation}%
s.t.
\begin{equation}
\left[ \mathbf{K}^{-1}\left( \mathbf{K}^{-1}\right) ^{T}\right] _{\text{diag}%
}=\mathbf{I}\text{.}  \label{Eq 2}
\end{equation}%
\end{subequations} We present a solution the this problem, with the
proof given in Appendix V.

\begin{lemm}
\label{MarkerLemmaApproxII}For $\mathbf{\Psi }_{A}=\mathbf{\Sigma }$ and $%
\mathbf{\Psi }_{B}=\gamma \mathbf{\Sigma }$ with $\mathbf{\Sigma
}$ given by (\ref{Eq Lemma3 Revised}), the optimal rotation matrix
$\mathbf{K}$ for the
problem in\ (\ref{Eq 1}) and (\ref{Eq 2}) is given by $\mathbf{K}%
_{opt}^{(\gamma )}$ in (\ref{Eq Theorem 1}).
\end{lemm}

The above lemma means that, given the power matrices $\mathbf{\Psi
}_{A}$ and $\mathbf{\Psi }_{B}$ obtained from AS-I, it is
impossible to find a more power-efficient $\mathbf{K}$ than the
unitary one given by (\ref{Eq Theorem 1}).

We next relax the second constraint. With $\mathbf{K}$ given by
(\ref{Eq
Theorem 1}), we optimize the power matrices $\mathbf{\Psi }_{A}$ and $%
\mathbf{\Psi }_{B}$ without confining to $\mathbf{\Psi }_{B}=\gamma \mathbf{%
\Psi }_{A}$. The corresponding WSR problem is written as
\begin{subequations}
\begin{equation}
\max_{\mathbf{\Psi }_{A},\text{ }\mathbf{\Psi }_{B}}\left\{ \alpha
R_{A,UL}^{EDA}+\left( 1-\alpha \right) R_{B,UL}^{EDA}\right\}
\label{Eq Objective}
\end{equation}%
subject to%
\begin{equation}
\text{Tr}\left( \mathbf{U}_{A}\mathbf{\Sigma }_{A}^{-2}\mathbf{U}_{A}^{T}%
\mathbf{K}_{opt}^{(\gamma )}\mathbf{\Psi }_{A}^{2}\left( \mathbf{K}%
_{opt}^{(\gamma )}\right) ^{T}+\mathbf{U}_{B}\mathbf{\Sigma }_{B}^{-2}%
\mathbf{U}_{B}^{T}\mathbf{K}_{opt}^{(\gamma )}\mathbf{\Psi
}_{B}^{2}\left( \mathbf{K}_{opt}^{(\gamma )}\right) ^{T}\right)
\leq P_{T}. \label{Eq Constraint}
\end{equation}%
\end{subequations}
The objective function in (\ref{Eq Objective}) is concave in $\mathbf{\Psi }%
_{A}$ and $\mathbf{\Psi }_{B}$ if we set the term $\frac{\mathbf{\Psi }%
_{A}^{2}}{\mathbf{\Psi }_{A}^{2}+\mathbf{\Psi }_{B}^{2}}$ of
$R_{A,UL}^{EDA}$
in (\ref{Eq RateA UL original}), and $\frac{\mathbf{\Psi }_{B}^{2}}{\mathbf{%
\Psi }_{A}^{2}+\mathbf{\Psi }_{B}^{2}}$ of $R_{B,UL}^{EDA}$ in
(\ref{Eq RateB UL original}), to be pre-determined constant
matrices $\mathbf{\theta } $ and $\mathbf{I}-\mathbf{\theta }$,
respectively. The solution can then be found by recursively
solving (\ref{Eq Objective}) by fixing $\mathbf{\theta }
$ (which is a convex optimization problem), and then updating $\mathbf{%
\theta }$ using the new solution of $\mathbf{\Psi }_{A}$ and $\mathbf{\Psi }%
_{B}$. The details are tedious and thus omitted here.

We summarize approximate solution II as follows:

\texttt{\ \ \ \ \ \ Approximate Solution II} \texttt{\
\begin{align*}
& \ \ \ \ \ \ \text{Given the }\gamma \text{ and
}\mathbf{K}_{opt}^{\left( \gamma \right) }\ \text{from Approxi.
Solution I}\ \ \ \ \ \ \ \ \ \ \ \ \ \
\ \ \ \ \ \  \\
& \ \ \ \ \ \ \ \ \ \ \text{solve the problem in (\ref{Eq Objective}) and (%
\ref{Eq Constraint})}
\end{align*}%
}
\section{Numerical Results}

In this section, we provide numerical results to evaluate the
performance of the proposed EDA-PNC scheme for MIMO TWRCs. In
simulation, we always assume
that the relay SNR and the average per-user SNR are identical, i.e., $%
SNR_{R}=$ $SNR$. The results presented below are obtained by
averaging over 1,000 channel realizations.

\subsection{Achievable Sum-Rates of MIMO TWRCs with $n_{T}\geq n_{R}=2$}

Here, we present the numerical results for real-valued MIMO TWRCs with $%
n_{T}=n_{R}=2$. The coefficients in the channel matrices are
independently drawn from $\mathcal{N}(0,1)$. The optimal rotation
matrix $\mathbf{K}$ and the optimal power matrices $\mathbf{\Psi
}_{A}$ and $\mathbf{\Psi }_{B}$ are found by utilizing the
exhaustive search method. The achievable sum-rate of the proposed
EDA-PNC scheme is plotted in Fig. \ref{Fig_SumRate_2by2}. The
sum-capacity UB of the MIMO TWRC, the achievable sum-rate UBs of
the ANC and DF-NC schemes, as well as the achievable sum-rate of
the naive EDA-PNC (with $\mathbf{K=I}$) scheme, are also included
for comparison. In the high SNR region, we observe that the gap
between the achievable sum-rate of the proposed EDA-PNC scheme and
the sum-capacity UB of the MIMO TWRC is very small, e.g., less
than 0.3 bit/Sec/Hz in spectral efficiency, or less than 0.4 dB in
power efficiency, at a SNR greater than 15 dB. We also see that
the proposed EDA-PNC scheme significantly outperforms the ANC,
DF-NC and the naive EDA scheme. Specifically, the ANC scheme
suffers from a significant power loss of about 3-4 dB compared
with the proposed EDA-PNC scheme. The DF-NC scheme suffers from a
severe multiplexing loss, as the slope of its performance curve is
nearly halved compared to the other schemes. In the low SNR
region, we observe that the DF-NC scheme almost achieves the
sum-capacity UB, and that the proposed EDA-PNC scheme is inferior
to the DF-NC scheme. This is due to the inherited disadvantage of
nested lattice codes in the low SNR region \cite{NamIT09}.

Next, we show the numerical result of the proposed EDA-PNC scheme
for MIMO
TWRCs with $n_{R}=2$ and $n_{T}=2,3,4$. In Fig. \ref%
{Fig_SumRate_nR2_Various_nT}, two performance curves of the
proposed EDA-PNC scheme are illustrated. One is based on the
exhaustive search method, and the other is based on the
approximate solution II (AS-II) method developed in Section VI.
The sum-capacity UBs of the MIMO TWRCs and the performance curves
of the DF-NC scheme are also plotted. In the medium-to-high SNR
region, the gap between the proposed EDA-PNC scheme and the
sum-capacity UB of the MIMO TWRC diminishes as $n_{T}$ increases.
This agrees well with the asymptotic optimality of EDA-PNC, as
stated in Theorem \ref{Theorem4}. We also see that there is a tiny
gap between the optimal EDA-PNC curve (obtained from the
exhaustive search) and the one based on AS-II. This implies that
the proposed AS-II algorithm is nearly optimal for $n_{R}$ $=2$.

\subsection{Achievable Rates of MIMO TWRCs with $n_{T}\geq n_{R}=4$}

Now, we consider complex-valued MIMO TWRCs with $n_{T}\geq
n_{R}=4$. The channel coefficients are now independently drawn
from $\mathcal{CN}$(0,1). In this case, the complexity of
exhaustive search in finding the the optimal EDA precoder is
prohibitively high. Thus, we confine our results to the
approximate solutions developed in Section VI.

\subsubsection{Achievable Sum-Rates}

In Fig. \ref{Fig_AchievableRate_MIMO_Nr4_VariousSchemes}, we plot
the achievable sum-rate of the proposed EDA-PNC scheme with
$n_{T}=n_{R}=4$. This figure also includes the performance curves
of the other schemes considered in Fig. \ref{Fig_SumRate_2by2}.
The only difference is that the AS-II algorithm is used in
plotting the performance curve of the proposed
EDA-PNC scheme. Comparing Fig. \ref%
{Fig_AchievableRate_MIMO_Nr4_VariousSchemes} with Fig. \ref{Fig_SumRate_2by2}%
, we see that the relative performance trends of these schemes are
quite similar, except that the gap between the proposed EDA-PNC
and the capacity UB is slightly larger (about 1.4 dB in power
efficiency in the high SNR region) in Fig.
\ref{Fig_AchievableRate_MIMO_Nr4_VariousSchemes}. We conjecture
that this performance degradation is mainly due to the
sub-optimality of AS-II. We will seek for the possibility of
improving AS-II in our future work.

In Fig. \ref{Fig_AchievableRate_MIMO_Nr4}, we further study the impact of $%
n_{T}$ on the achievable sum-rate in the case of $n_{R}$ $=4$.
Similar to Fig. \ref{Fig_SumRate_nR2_Various_nT}, we see that the
proposed EDA-PNC scheme asymptotically approaches the capacity UB
as $n_{T}$ increases. It is also worth mentioning that, for
$n_{T}=8$ and $n_{R}=4$, the proposed EDA-PNC scheme can increase
the spectral efficiency by more than 50\% relative to the DF-NC
scheme, at a practical SNR level (e.g., SNR=15 dB). In
addition, we compare the performance of AS-I and AS-II algorithms in Fig. %
\ref{Fig_AchievableRate_MIMO_Nr4}. We see that AS-II always
slightly outperforms AS-I. For this reason, we only include the
performance curves of AS-II in the other figures presented in this
paper.

\subsubsection{Achievable Rate-Regions}

We next show the achievable rate-region of the proposed EDA-PNC
scheme (based on AS-II). The results for the case of
$n_{T}=n_{R}=4$ is shown in Fig. \ref{Fig_RateRegion_nT4nR4}, at
SNR = 0, 10, 15, 25 dB. We also include the rate-regions of the
capacity UB, the DF-NC scheme and the naive EDA-PNC scheme.
Clearly, the proposed scheme achieves a significantly larger
rate-region relative to the DF-NC scheme and the naive EDA-PNC
scheme, at a medium-to-high SNR. For a SNR of 15 dB, the proposed
EDA-PNC scheme outperforms the DF-NC scheme, whereas the naive
EDA-PNC scheme is worse than the DF-NC scheme, for the entire
rate-region. Compared to the naive EDA precoding, the performance
gain achieved by the proposed EDA precoding is significant. For
low SNRs, e.g., SNR = 0 dB, the\ achievable rate-region of DF-NC
is very close to the capacity outer bound of the MIMO TWRC and is
better than that of the EDA-PNC scheme. This is in agreement with
the
observations in Figs. \ref{Fig_SumRate_2by2}-\ref%
{Fig_AchievableRate_MIMO_Nr4}.

Finally, in Fig. \ref{Fig_RateRegion_nT8nR4}, we plot the
achievable rate-region of the proposed EDA-PNC scheme with
$n_{T}=8$ and $n_{R}=4$. Comparing to Fig.
\ref{Fig_RateRegion_nT4nR4}, we observe that the gap between the
achievable rate-region of the proposed EDA-PNC scheme and the
capacity outer bound of MIMO TWRC becomes smaller for the entire
SNR range. This agrees well with Theorem 2.

In summary, the results shown in Fig. \ref{Fig_SumRate_2by2}-\ref%
{Fig_RateRegion_nT8nR4} clearly demonstrates the benefits of the
proposed EDA-PNC scheme for MIMO TWRCs.

\section{Conclusions}
In this paper, we proposed an EDA-PNC scheme to approach the
capacity of a MIMO TWRC. The proposed EDA precoder efficiently
creates $n_{R}$ aligned parallel channels for the two users, which
provides a platform to perform multi-stream PNC. In such a manner,
the benefits of PNC can now be exploited in a MIMO two-way relay
system. We derived an achievable rate of the proposed EDA-PNC
scheme and showed that, as $n_{T}/n_{R}$ increases (towards
infinity), the proposed EDA-PNC scheme approaches the capacity
upper bound of a MIMO TWRC. For a finite $n_{T}$, numerical
results demonstrated that there is only a marginal gap between the
achievable rate of the proposed scheme and the capacity upper
bound, and the proposed scheme clearly outperforms the existing
benchmark schemes. It is worth mentioning that the discussions in
this paper is limited to the situation of $n_{T}\geq n_{R}$. The
extension of this work to the case of $n_{T}<n_{R}$ requires a
dimension reduction method and is of interest for future work.

\section*{Appendix I Treatment for a Complex-Valued Model}
The results of this paper derived based on a real-valued system
model can be readily extended to the case of a complex-valued
model. The key observation is that every complex-valued system
model can be equivalently expressed in a real-valued form.

For example, suppose that the uplink channel model in (\ref%
{Eq_SystemModel_Uplink}) is complex-valued. It can be equivalently
expressed
in a real-valued form as%
\begin{equation}
\left[
\begin{array}{c}
\mathfrak{R}\left( Y_{R}\right) \\
\mathfrak{I}\left( Y_{R}\right)%
\end{array}%
\right] =\sum\limits_{m\in \left\{ A,B\right\} }\left[
\begin{array}{cc}
\mathfrak{R}\left( \mathbf{H}_{m,R}\right) & -\mathfrak{I}\left( \mathbf{H}%
_{m,R}\right) \\
\mathfrak{I}\left( \mathbf{H}_{m,R}\right) & \mathfrak{R}\left( \mathbf{H}%
_{m,R}\right)%
\end{array}%
\right] \left[
\begin{array}{c}
\mathfrak{R}\left( X_{m}\right) \\
\mathfrak{I}\left( X_{m}\right)%
\end{array}%
\right] +\left[
\begin{array}{c}
\mathfrak{R}\left( Z_{R}\right) \\
\mathfrak{I}\left( Z_{R}\right)%
\end{array}%
\right]
\end{equation}%
where $\mathfrak{R}\left( \mathbf{\cdot }\right) $ and
$\mathfrak{I}\left( \mathbf{\cdot }\right) $ denote the real part
and imaginary part of a complex-valued matrix (or a vector),
respectively.

It is noteworthy that the above relationship also applies to the
downlink channel model (\ref{Eq_SystemModel_Downlink}). In this
way, the results obtained for the real-valued system are directly
applicable to a complex-valued system.

\section*{Appendix II Proof of Theorem \protect\ref{Theorem2}}

Here we only provide a sketch of the proof. We refer the
interested readers to \cite{NamIT09} (cf., proof of Th.1 in
\cite{NamIT09}) for more details.

\subsection{Uplink Achievable Rate-Pair}

Recall from (\ref{Eq_Relay_Signal_Multiply_K Step2}) that the
$n_{R}$ aligned eigen-modes (sub-channels) created by EDA
precoding can be written
in an entry-by-entry form as%
\begin{equation}
y_{R,i}\left[ l\right] =\Psi _{A}\left( i,i\right) c_{A,i}\left[
l\right] +\Psi _{B}\left( i,i\right) c_{B,i}\left[ l\right]
+z_{R,i}\left[ l\right] ,l=1,\cdots ,n,i=1,\cdots ,n_{R}.
\end{equation}%
where $y_{n_{R},i}\left[ l\right] $ (or $z_{n_{R},i}\left[
l\right] $)
represents the $i$th entry of $\widetilde{Y}_{R}\left[ l\right] $ (or $%
\widetilde{Z}_{R}\left[ l\right] $).

\subsubsection{Encoding}

The construction of nested lattice codes for each sub-channel $i$
follows exactly from \cite{NamIT09}. Let $\mathcal{C}_{m,i}$,
$m\in \left\{
A,B\right\} $, be the codebook of user $m$ for the $i$th sub-channel, and 2$%
^{nR_{m,i}}$ be the size of $\mathcal{C}_{m,i}$. To deliver a
message in the $i$th sub-channel, user $m$ chooses a codeword
$W_{m,i}\in \mathcal{C}_{m,i}$ associated with the message. After
a random dithering and a module-lattice operation \cite{NamIT09},
a length-$n$ signal sequence
\begin{equation*}
\mathbf{c}_{m,i}=\left[ c_{m,i}\left[ 1\right] ,\cdots ,c_{m,i}\left[ n%
\right] \right] ,
\end{equation*}%
is generated which will be transmitted in the $i$th sub-channel.
The above encoding operation is performed for all $n_{R}$
sub-channels.

\subsubsection{Decoding (the bin-index) at the Relay}

%The relay's operation is depicted in Fig. \ref{Fig_EDACPNC_RelayDecoder}.
Upon receiving $\left\{ \widetilde{Y}_{R}\left[ l\right] \right\} _{l=1}^{n}$%
, the relay computes the so-called \textquotedblleft
bin-index\textquotedblright\ $T_{i}$ instead of $W_{A,i}$ and
$W_{B,i}$ for each sub-channel $i$, $i=1,\cdots ,n_{R}$. (See the
definition of bin-index in \cite{NamIT09}). From Theorem 3 in
\cite{NamIT09}, the error probability
of recovering the bin-index $T_{i}$ at the relay is arbitrarily small as $%
n\rightarrow \infty $ if
\begin{subequations}
\begin{eqnarray}
R_{A,i} &\leq &\frac{1}{2}\left[ \log \left( \frac{\Psi _{A}\left(
i,i\right) ^{2}}{\Psi _{A}\left( i,i\right) ^{2}+\Psi _{B}\left(
i,i\right)
^{2}}+\Psi _{A}\left( i,i\right) ^{2}\right) \right] ^{+},\text{ } \\
R_{B,i} &\leq &\frac{1}{2}\left[ \log \left( \frac{\Psi _{B}\left(
i,i\right) ^{2}}{\Psi _{A}\left( i,i\right) ^{2}+\Psi _{B}\left(
i,i\right) ^{2}}+\Psi _{B}\left( i,i\right) ^{2}\right) \right]
^{+},
\end{eqnarray}%
\end{subequations}
where $i=1,\cdots ,n_{R}.$

Since the aligned sub-channels are orthogonal to each other, the rate-pair ($%
R_{A}$, $R_{B}$) with which all the bin-indices $\left\{
T_{i}\right\} _{i=1}^{n_{R}}$ can be recovered correctly is given
by (\ref{Eq RateA UL original}) and (\ref{Eq RateB UL original}).

\subsection{Downlink Achievable Rate-Pair}

\subsubsection{Relay's Encoding}

Define a \textquotedblleft super bin-index\textquotedblright\ as $T$ $%
\mathbf{\triangleq \lbrack }T_{1},T_{2},\cdots
,T_{n_{R}}\mathbf{]}$ and
assume that $T$ is recovered correctly by the relay. Also, assume that $%
R_{A}\geq R_{B}$\footnote{%
The derivation for the case of $R_{A}<R_{B}$ will be similar.}. We generate 2%
$^{nR_{A}}$ $n_{R}$-by-$n$ codeword matrices with each column
drawn independently from a multi-variant Gaussian distribution
with zero mean and
covariance $\mathbf{Q}_{R}$. This forms a rate-$R_{A}$ codebook $\mathcal{C}%
_{R}$ (whose generation is independent of the codebooks $\left\{ \mathcal{C}%
_{m,1},\cdots ,\mathcal{C}_{m,n_{R}}\right\} $ used in the uplink
phase). The codebook $\mathcal{C}_{R}$ is employed to map each
super bin-index $T$ into a codeword in $\mathcal{C}_{R}$. Denote
by $\mathbf{X}_{R}\left(
T\right) $ the codeword in $\mathcal{C}_{R}$ mapped to $T$. Then, $\mathbf{X}%
_{R}\left( T\right) $ is transmitted over the $n_{R}$ antennas at
the relay.

\subsubsection{Decoding of the Two Users}

%The decoding at the two users in the downlink phase are illustrated in Fig. %
%\ref{Fig_EDACPNC_DestinationDecoder}.
Upon receiving $\mathbf{Y}_{A}$, user $A$ decodes $T_{A}$, by finding in $%
\mathcal{C}_{A}^{DL}$ a codeword that is jointly typical with $\mathbf{Y}%
_{A} $. Here, $\mathcal{C}_{A}^{DL}$ is constructed by selecting
the codewords in $\mathcal{C}_{R}$ corresponding to $\left[
W_{A,1},\cdots ,W_{A,n_{R}}\right] $ (which are perfectly known to
user $A)$. Note that the cardinality of $\mathcal{C}_{A}^{DL}$ is
$2^{nR_{B}}$ \cite{NamIT09}. From
the argument of random coding and jointly typical decoding \cite%
{CoverTextBook}, we have $\Pr \left[ T_{A}\neq T\right] \rightarrow 0$ as $%
n\rightarrow \infty $ if
\begin{subequations}
\begin{equation}
R_{B}\leq R_{B,DL}^{EDA}=\frac{1}{2}\log \det \left( \mathbf{I}+\mathbf{H}%
_{R,A}\mathbf{Q}_{R}\mathbf{H}_{R,A}^{T}\right) .  \label{Eq R_B
EDA DL}
\end{equation}%
With $T=\left[ T_{1},\cdots ,T_{n_{R}}\right] $ and $\left[
W_{A,1},\cdots
,W_{A,n_{R}}\right] $, user $A$ can uniquely determine the messages of user $%
B$ using the method described in \cite{NamIT09}.

Similarly, user $B$ can reliably determine the messages of user A if%
\begin{equation}
R_{A}\leq R_{A,DL}^{EDA}=\frac{1}{2}\log \det \left( \mathbf{I}+\mathbf{H}%
_{R,B}\mathbf{Q}_{R}\mathbf{H}_{R,B}^{T}\right) .  \label{Eq R_A
EDA DL}
\end{equation}%
\end{subequations}
Combining (\ref{Eq RateA UL original}), (\ref{Eq RateB UL original}), (\ref%
{Eq R_B EDA DL}) and (\ref{Eq R_A EDA DL}), we complete the proof of Theorem %
\ref{Theorem2}.

\section*{Appendix III Proof of Theorem \protect\ref{Theorem4}}

\begin{proof}[Proof of Theorem \protect\ref{Theorem4}]
Since the downlink rate-pair of the EDA-PNC scheme is identical to
that of the capacity UB, we only need to consider the uplink
rate-pair.
Specifically, we need to show that%
\begin{equation}
R_{m,UL}^{UB}-R_{m,UL}^{EDA}\overset{\text{P}}{\rightarrow }0\text{ for }%
n_{T}\rightarrow \infty \text{, }m\in \left\{ A,B\right\} \text{.}
\label{Eq AsymptoticallyNoGap}
\end{equation}

Clearly, $R_{m,UL}^{UB}$ and $R_{m,UL}^{EDA}$ are continuous functions of $%
\mathbf{H}_{m,R}\mathbf{H}_{m,R}^{T}$. From the property of
convergence in probability (cf., Theorem 4, pp. 261 of
\cite{Rohatgi}), to prove (\ref{Eq
AsymptoticallyNoGap}), it suffices to show that, if%
\begin{equation}
\frac{1}{n_{T}}\mathbf{H}_{m,R}\mathbf{H}_{m,R}^{T}\rightarrow \mathbf{I}%
\text{, as }n_{T}\rightarrow \mathbf{\infty ,}  \label{A2}
\end{equation}%
then
\begin{equation}
R_{m,UL}^{UB}-R_{m,UL}^{EDA}\rightarrow 0\text{, for
}n_{T}\rightarrow \infty .  \label{A3}
\end{equation}%
From (\ref{Eq_RateA_UB_MIMO}) and (\ref{Eq_RateB_UB_MIMO}), we
have
\begin{equation}
R_{m,UL}^{UB}=\frac{1}{2}\log \det \left( \mathbf{I}+\mathbf{H}_{m,R}\mathbf{%
Q}_{m}\mathbf{H}_{m,R}^{T}\right) .  \label{A4}
\end{equation}%
Let $P_{m,T}$ be the power allocated to user $m$, $m\in \left\{ A,B\right\} $%
. With (\ref{A2}), it can be shown that as $n_{T}\rightarrow
\mathbf{\infty } $, the optimal $\mathbf{Q}_{m}$ takes the form of
\begin{equation}
\mathbf{Q}_{m}=\frac{P_{m,T}}{n_{R}n_{T}}\mathbf{H}_{m,R}^{T}\mathbf{H}%
_{m,R}.  \label{A5}
\end{equation}%
Thus, as $n_{T}\rightarrow \mathbf{\infty }$, we obtain
\begin{subequations}
\begin{eqnarray}
R_{m,UL}^{UB} &=&\frac{1}{2}\log \det \left( \mathbf{I+}\frac{P_{m,T}}{%
n_{R}n_{T}}\mathbf{H}_{m,R}\mathbf{H}_{m,R}^{T}\mathbf{H}_{m,R}\mathbf{H}%
_{m,R}^{T}\right)  \label{Step a} \\
&=&\frac{1}{2}n_{R}\log n_{T}+\frac{1}{2}\log \det \left[ \frac{1}{n_{T}}%
\mathbf{I+}\frac{P_{m,T}}{n_{R}}\left( \frac{1}{n_{T}}\mathbf{H}_{m,R}%
\mathbf{H}_{m,R}^{T}\right) \left( \frac{1}{n_{T}}\mathbf{H}_{m,R}\mathbf{H}%
_{m,R}^{T}\right) \right]  \notag \\
&=&\frac{1}{2}n_{R}\log n_{T}+\frac{1}{2}\log \det \left[ \frac{P_{m,T}}{%
n_{R}}\mathbf{I}\right] +o\left( 1\right)  \label{Step b} \\
&=&\frac{n_{R}}{2}\log \left( \frac{n_{T}P_{m,T}}{n_{R}}\right)
+o\left( 1\right) ,  \label{A6}
\end{eqnarray}%
\end{subequations}
where (\ref{Step a}) follows by substituting (\ref{A5}) into
(\ref{A4}), and (\ref{Step b}) follows from (\ref{A2}).

Now, consider the EDA precoder
\begin{equation}
\mathbf{F}_{m}=\mathbf{H}_{m,R}^{T}\left( \mathbf{H}_{m,R}\mathbf{H}%
_{m,R}^{T}\right) ^{-1}\mathbf{K\Psi }_{m}\,\text{, }m\in \left\{
A,B\right\} \text{.}  \label{A7}
\end{equation}%
We choose
\begin{equation}
\mathbf{K=I}\text{ and }\mathbf{\Psi }_{m}=\sqrt{\frac{n_{T}P_{m,T}}{n_{R}}}%
\mathbf{I}.  \label{A8}
\end{equation}%
The power constraint is asymptotically met, i.e.,
\begin{equation}
\text{Tr}\left\{ \mathbf{F}_{m}\mathbf{F}_{m}^{T}\right\}
=\text{Tr}\left\{
\left( \mathbf{H}_{m,R}\mathbf{H}_{m,R}^{T}\right) ^{-1}\frac{n_{T}P_{m,T}}{%
n_{R}}\right\} \overset{\text{P}}{\rightarrow }P_{m,T}\text{, for }%
n_{T}\rightarrow \infty .  \label{A9}
\end{equation}%
The choice of $\mathbf{K}$ and $\mathbf{\Psi }_{m}$ in (\ref{A8})
is in no sense optimal. However, we will show that this suboptimal
choice is sufficient to prove (\ref{A3}). To see this, the uplink
achievable rate of
the EDA-PNC scheme is given by %\begin{subequations}
\begin{eqnarray}
R_{m,UL}^{EDA} &=&\frac{1}{2}\sum_{i=1}^{n_{R}}\left[ \log \left(
\frac{\Psi _{m}^{2}\left( i,i\right) }{\Psi _{A}^{2}\left(
i,i\right) +\Psi _{B}^{2}\left( i,i\right) }+\Psi _{m}^{2}\left(
i,i\right) \right) \right]
^{+}  \label{A10a} \\
&=&\frac{n_{R}}{2}\log \left( \frac{n_{T}P_{m,T}}{n_{R}}\right)
+o(1)\text{, as }n_{T}\rightarrow \infty ,  \label{A10b}
\end{eqnarray}%
where (\ref{A10a}) follows from Theorem 2 and (\ref{A10b}) is from (\ref{A8}%
) together with the fact that $\frac{\Psi _{m}^{2}\left(
i,i\right) }{\Psi
_{A}^{2}\left( i,i\right) +\Psi _{B}^{2}\left( i,i\right) }\leq 1$, for $%
i=1,\cdots ,n_{R}$.

Combining (\ref{A6}) and (\ref{A10b}), we arrive at (\ref{A3}),
which completes the proof of Theorem \ref{Theorem4}.
\end{proof}

\section*{Appendix IV Proof of Lemma \protect\ref{Lemma3}}

\begin{proof}[Proof of Lemma \protect\ref{Lemma3}]
The objective function (\ref{Eq Optimization ULRegion PlusPlus})
is jointly concave in $\left\{ \Sigma \left( 1,1\right) ,\cdots
,\Sigma \left( n_{R},n_{R}\right) \right\} $. The Lagrangian of
problem (\ref{Eq Optimization ULRegion PlusPlus}) is given by
\begin{eqnarray}
&&L\left( \lambda ,\Sigma \left( 1,1\right) ,\cdots ,\Sigma \left(
n_{R},n_{R}\right) \right)   \notag \\
&=&\frac{\alpha }{2}\dsum\limits_{i\in S_{A}}\left[ \log \left( \frac{1}{%
1+\gamma ^{2}}+\Sigma \left( i,i\right) ^{2}\right) \right]
+\frac{1-\alpha }{2}\dsum\limits_{i\in S_{B}}\left[ \log \left(
\frac{\gamma ^{2}}{1+\gamma
^{2}}+\gamma ^{2}\Sigma \left( i,i\right) ^{2}\right) \right]   \notag \\
&&-\lambda \dsum\limits_{i=1}^{k}\Lambda _{G\left( \gamma \right)
}\left( i,i\right) \Sigma \left( i,i\right) ^{2}  \label{Eq
Lagrangian}
\end{eqnarray}%
where $\lambda $ is a non-negative scalar. The partial derivative
of the Lagrangian in (\ref{Eq Lagrangian}) with respect to each
$\Sigma \left( i,i\right) ^{2}$ is given by
\begin{equation}
\frac{\partial L}{\partial \left( \Sigma \left( i,i\right) ^{2}\right) }%
=\left\{
\begin{array}{cc}
\frac{1}{2}\frac{1}{\frac{1}{1+\gamma ^{2}}+\Sigma \left( i,i\right) ^{2}}%
-\lambda \Lambda _{G\left( \gamma \right) }\left( i,i\right) , & \text{if }%
i\in S_{A}\text{ and }i\in S_{B} \\
\frac{1}{2}\frac{\alpha }{\frac{1}{1+\gamma ^{2}}+\Sigma \left(
i,i\right)
^{2}}-\lambda \Lambda _{G\left( \gamma \right) }\left( i,i\right) , & \text{%
if }i\in S_{A}\text{ and }i\notin S_{B} \\
\frac{1}{2}\frac{1-\alpha }{\frac{1}{1+\gamma ^{2}}+\Sigma \left(
i,i\right)
^{2}}-\lambda \Lambda _{G\left( \gamma \right) }\left( i,i\right) , & \text{%
if }i\notin S_{A}\text{ and }i\in S_{B} \\
0, & \text{if }i\notin S_{A}\text{ and }i\notin S_{B}%
\end{array}%
\right.   \notag
\end{equation}%
where $i=1,..,n_{R}.$ The Karush-Kuhn-Tucker condition is%
\begin{equation}
\frac{\partial L}{\partial \left( \Sigma \left( i,i\right) ^{2}\right) }%
\left\{
\begin{array}{c}
=0,\text{ if }\Sigma \left( i,i\right) >0 \\
\leq 0,\text{ if }\Sigma \left( i,i\right) =0%
\end{array}%
\right.
\end{equation}%
which yields the solution (\ref{Eq Lemma3}).
\end{proof}

\section*{Appendix V Proof of Lemma \protect\ref{MarkerLemmaApproxII}}

\ Without loss of generality, let the SVD of $\mathbf{K\Sigma }$ be%
\begin{equation}
\mathbf{K\Sigma }=\mathbf{UDV}^{T}.  \label{Eq AppendII 2}
\end{equation}%
where the diagonal elements of $\mathbf{\Sigma }$ and $\mathbf{D}$
are both
arranged in the descending order. From (\ref{Eq AppendII 2}), the rank of $%
\mathbf{D}$ is the same as $\mathbf{\Sigma }$ (as $\mathbf{K}$, $\mathbf{U}$%
, and $\mathbf{V}$ are all of full rank). This implies that, if
$\Sigma \left( i,i\right) $ $=0$ for any index $i$, then $D\left(
i,i\right) $ $=0$.

Let us first consider that $\mathbf{\Sigma }$ has full rank. We
will relax this constraint later. Using
(\ref{Eq_Trace_inequlity}), we obtain
\begin{equation*}
\text{Tr}(\mathbf{G}(\gamma
)\mathbf{UD}^{2}\mathbf{U}^{T})=\text{Tr}\left(
\mathbf{U}_{G\left( \gamma \right) }\mathbf{\Lambda }_{G\left(
\gamma
\right) }\mathbf{U}_{G\left( \gamma \right) }^{T}\mathbf{UD}^{2}\mathbf{U}%
^{T}\right) \geq \text{Tr}\left( \mathbf{\Lambda }_{G\left( \gamma
\right) }\cdot \mathbf{D}^{2}\right)
\end{equation*}%
where the diagonal entries of $\mathbf{\Lambda }_{G\left( \gamma
\right) }$
are arranged in the ascending order, and the equality holds when $\mathbf{U}=%
\mathbf{U}_{G\left( \gamma \right) }$.

Then, the optimization problem in (\ref{Eq 1}) and (\ref{Eq 2})
can be
expressed as%
\begin{subequations}
\begin{equation}
\min_{\mathbf{D},\mathbf{V}}\text{Tr}\left( \mathbf{\Lambda
}_{G\left( \gamma \right) }\cdot \mathbf{D}^{2}\right)   \label{Eq
AppendII 3}
\end{equation}%
s.t.%
\begin{equation}
\left[ \mathbf{VD}^{-2}\mathbf{V}^{T}\right] _{\text{diag}}=\mathbf{\Sigma }%
^{-2}.  \label{Eq AppendII 4}
\end{equation}%
\end{subequations}
Note that the diagonal elements of $\mathbf{\Sigma }^{-2}$ and $\mathbf{D}%
^{-2}$ are both arranged in the ascending order. Denote the
diagonal entries
of $\mathbf{\Sigma }^{-2}$ by $\left[ \sigma _{1},\cdots \sigma _{n_{R}}%
\right] $ and those of $\mathbf{D}^{-2}$ as $\left[ d_{1},\cdots ,d_{n_{R}}%
\right] $. From Th. 4.3.32 of \cite{HornTextbook}, for any $\left[
d_{1},\cdots ,d_{n_{R}}\right] $ majorized by $\left[ \sigma
_{1},\cdots \sigma _{n_{R}}\right] $, there always exists a
unitary matrix $\mathbf{V}$ satisfying (\ref{Eq AppendII 4}). \
Therefore, the optimization problem
specified in (\ref{Eq AppendII 3}) and (\ref{Eq AppendII 4}) becomes%
\begin{equation}
\min_{d_{1},\cdots ,d_{n_{R}}}\sum_{i=1}^{n_{R}}\frac{\Lambda
_{G\left( \gamma \right) }\left( i,i\right) }{d_{i}}  \label{Eq
AppendII 5}
\end{equation}%
subject to the majorization constraint as \cite{HornTextbook}%
\begin{eqnarray}
d_{i} &\geq &0\text{ }\forall \text{ }i\in \left\{ 1,\cdots
,n_{R}\right\}
\label{Eq AppendII 6} \\
d_{1} &\leq &d_{2}\leq \cdots \leq d_{n_{R}}  \notag \\
\sum_{i=1}^{t}d_{i} &\leq &\sum_{i=1}^{t}\sigma _{i},t=1,\cdots
,n_{R}-1,
\notag \\
\sum_{i=1}^{n_{R}}d_{i} &=&\sum_{i=1}^{n_{R}}\sigma _{i}.  \notag
\end{eqnarray}

We next show that, with $\mathbf{\Sigma }$ given by (\ref{Eq Lemma3 Revised}%
), the solution to the optimization problem (\ref{Eq 1}) (\ref{Eq
2}) is
given by%
\begin{equation}
d_{i}=\sigma _{i},\forall \text{ }i=1,...,n_{R}.  \label{Eq 7}
\end{equation}

To prove (\ref{Eq 7}), we need some facts, as detailed below.

\begin{fact}
\label{MarkerProp1}For any $i,j\in \left\{ 1,\cdots ,n_{R}\right\} $ with $%
i<j$, we have%
\begin{equation}
\frac{\Lambda _{G}\left( i,i\right) }{\sigma _{i}^{2}}\geq
\frac{\Lambda _{G}\left( j,j\right) }{\sigma _{j}^{2}}.  \label{Eq
4}
\end{equation}
\end{fact}

\begin{proof}[Proof of Fact 2]
From (\ref{Eq Lemma3 Revised}), we have%
\begin{equation}
\sigma _{i}=\frac{1}{\left[ \left( \frac{1}{2\lambda \Lambda
_{G\left(
\gamma \right) }\left( i,i\right) }-\frac{1}{1+\gamma ^{2}}\right) ^{+}%
\right] }  \label{Eq 5}
\end{equation}%
where $\lambda >0$. Note that $\sigma _{i}>0$ for all $i=1,...,n_{R}$ (as $%
\mathbf{\Sigma }$ is of full rank). Then,
\begin{equation}
\frac{\sigma _{i}^{2}}{\sigma _{j}^{2}}=\frac{\left(
\frac{1}{2\lambda
\Lambda _{G\left( \gamma \right) }\left( j,j\right) }-\frac{1}{1+\gamma ^{2}}%
\right) ^{2}}{\left( \frac{1}{2\lambda \Lambda _{G\left( \gamma
\right) }\left( i,i\right) }-\frac{1}{1+\gamma ^{2}}\right)
^{2}}\overset{\left( 1\right) }{\leq }\frac{\left(
\frac{1}{2\lambda \Lambda _{G\left( \gamma \right) }\left(
j,j\right) }\right) ^{2}}{\left( \frac{1}{2\lambda \Lambda
_{G\left( \gamma \right) }\left( i,i\right) }\right)
^{2}}\overset{\left( 2\right) }{\leq }\frac{\frac{1}{2\lambda
\Lambda _{G\left( \gamma \right) }\left( j,j\right)
}}{\frac{1}{2\lambda \Lambda _{G\left( \gamma \right) }\left(
i,i\right) }}=\frac{\Lambda _{G\left( \gamma \right) }\left(
i,i\right) }{\Lambda _{G\left( \gamma \right) }\left( j,j\right)
}, \label{Eq 6}
\end{equation}%
where steps $\overset{\left( 1\right) }{\leq }$ and
$\overset{\left( 2\right) }{\leq }$ are both from
$\frac{1}{2\lambda \Lambda _{G\left( \gamma \right) }\left(
j,j\right) }\leq \frac{1}{2\lambda \Lambda _{G\left( \gamma
\right) }\left( i,i\right) }$ (as the diagonal elements of \textbf{$\Lambda $%
}$_{G\left( \gamma \right) }$ are arranged in the ascending
order). This yields (\ref{Eq 4}).
\end{proof}

From (\ref{Eq AppendII 6}), we see that, for any $\left[
d_{1},\cdots
,d_{n_{R}}\right] $ in the feasible region, $d_{1}\leq \sigma _{1}$ and $%
d_{n_{R}}\geq \sigma _{n_{R}}$. Denote the objective function in
(\ref{Eq AppendII 5}) as
\begin{equation}
f\left( d_{1},\cdots ,d_{n_{R}}\right) =\dsum\limits_{i=1}^{n_{R}}\frac{%
\Lambda _{G\left( \gamma \right) }\left( i,i\right) }{d_{i}}.
\end{equation}%
For any $d_{1}\leq \sigma _{1}$ and $d_{n_{R}}\geq \sigma _{n_{R}}$, let $%
\varepsilon $ be a non-negative number satisfying%
\begin{equation}
d_{1}+\varepsilon \leq \sigma _{1}\text{ and
}d_{n_{R}}-\varepsilon \geq \sigma _{n_{R}}.  \label{xx}
\end{equation}

\begin{fact}
\begin{equation}
f\left( d_{1}+\varepsilon ,\cdots ,d_{n_{R}}-\varepsilon \right)
\leq f\left( d_{1},\cdots ,d_{n_{R}}\right)
\end{equation}%
where equality holds when $\varepsilon =0$.
\end{fact}

\begin{proof}[Proof of Fact 3]
We have%
\begin{align*}
& f\left( d_{1}+\varepsilon ,\cdots ,d_{n_{R}}-\varepsilon \right)
-f\left(
d_{1},\cdots ,d_{n_{R}}\right)  \\
& =\frac{\Lambda _{G\left( \gamma \right) }\left( 1,1\right) }{%
d_{1}+\varepsilon }+\frac{\Lambda _{G\left( \gamma \right) }\left(
n_{R},n_{R}\right) }{d_{n_{R}}-\varepsilon }-\left( \frac{\Lambda
_{G\left( \gamma \right) }\left( 1,1\right) }{d_{1}}+\frac{\Lambda
_{G\left( \gamma
\right) }\left( n_{R},n_{R}\right) }{d_{n_{R}}}\right)  \\
& =-\Lambda _{G\left( \gamma \right) }\left( 1,1\right) \frac{\varepsilon }{%
\left( d_{1}+\varepsilon \right) d_{1}}+\Lambda _{G\left( \gamma
\right) }\left( n_{R},n_{R}\right) \frac{\varepsilon }{\left(
d_{n_{R}}-\varepsilon
\right) d_{n_{R}}} \\
& \overset{(a)}{\leq }-\Lambda _{G\left( \gamma \right) }\left(
1,1\right) \frac{\varepsilon }{\left( d_{1}+\varepsilon \right)
^{2}}+\Lambda _{G\left( \gamma \right) }\left( n_{R},n_{R}\right)
\frac{\varepsilon }{\left(
d_{n_{R}}-\varepsilon \right) ^{2}} \\
& =\left( -\frac{\Lambda _{G\left( \gamma \right) }\left( 1,1\right) }{%
\left( d_{1}+\varepsilon \right) ^{2}}+\frac{\Lambda _{G\left(
\gamma \right) }\left( n_{R},n_{R}\right) }{\left(
d_{n_{R}}-\varepsilon \right)
^{2}}\right) \varepsilon  \\
& \overset{(b)}{\leq }\left( -\frac{\Lambda _{G\left( \gamma
\right) }\left( 1,1\right) }{\sigma _{1}^{2}}+\frac{\Lambda
_{G\left( \gamma \right) }\left(
n_{R},n_{R}\right) }{\sigma _{n_{R}}^{2}}\right) \varepsilon  \\
& \overset{(c)}{\leq }0
\end{align*}%
where step ($a$) is self-evident, step ($b$) follows from
(\ref{xx}) and
step ($c$) follows from (\ref{Eq 4}) in Fact 1. The equalities in steps ($a$%
)-($c$) hold when $\varepsilon =0$, which completes the proof.
\end{proof}

Fact 3 implies that the objective function $f\left(
d_{1}+\varepsilon ,\cdots ,d_{n_{R}}-\varepsilon \right) ,$ with
$\varepsilon $ constrained by
(\ref{xx}), is minimized when $\varepsilon =\sigma _{1}-d_{1}$ or $%
\varepsilon =\sigma _{n_{R}}-d_{n_{R}}$. Therefore, the optimum of
the problem in (\ref{Eq AppendII 5}) is achieved at either
$d_{1}=\sigma _{1}$
or $d_{n_{R}}=\sigma _{n_{R}}$. Without loss of generality, we assume that $%
d_{1}=\sigma _{1}$. Then, the dimension of the problem in (\ref{Eq
AppendII
5}) reduces from $n_{R}$ to $n_{R}-1$. Applying the same reasoning to this ($%
n_{R}-1$)-dimension problem, we can further show that
$d_{2}=\sigma _{2}$. Continuing this process, we eventually have
(\ref{Eq 7}), or equivalently,
\begin{equation*}
\mathbf{D}=\mathbf{\Sigma }.
\end{equation*}%
Therefore, from (\ref{Eq AppendII 2}) and the uniqueness of SVD, we obtain $%
\mathbf{V}=\mathbf{I}$ and $\widetilde{\mathbf{K}}_{opt}^{\left(
\gamma \right) }\mathbf{=U}_{G\left( \gamma \right)
}=\mathbf{K}_{opt}^{(\gamma )}.$

Next, consider that $\mathbf{\Sigma }$ does not have full rank.
Define
\begin{equation*}
\overline{\mathbf{\Sigma }}=\mathbf{\Sigma }+\delta
\sqrt{\mathbf{\Delta }}
\end{equation*}%
where $\delta $ is an arbitrary positive number, and
$\mathbf{\Delta }$ is a diagonal matrix with non-negative diagonal
elements. We can properly choose such a $\mathbf{\Delta }$ that:
(a) $\overline{\mathbf{\Sigma }}$ is of full
rank; (b) For a sufficiently small $\delta $, Fact 2 always holds for $%
\overline{\mathbf{\Sigma }}$ (and so does Fact 3). To this end, we choose%
\begin{equation}
\Delta \left( i,i\right) =0\text{, if }\Sigma \left( i,i\right)
\neq 0 \label{Eq Delta Case1}
\end{equation}%
and

\begin{equation}
\frac{\Lambda _{G}\left( i,i\right) }{\Delta ^{-2}\left(
i,i\right) }\geq \frac{\Lambda _{G}\left( j,j\right) }{\Delta
^{-2}\left( j,j\right) }\text{, if }\Sigma \left( i,i\right)
=\Sigma \left( j,j\right) =0,i<j\text{.} \label{Eq Delta Case2}
\end{equation}%
We verify Fact 2 for the above choice of $\mathbf{\Delta }$.
Noting that the diagonal entries of $\mathbf{\Sigma }$ are
arranged in descending order, we only need
to consider three cases: (i) $\Sigma \left( i,i\right) $ \TEXTsymbol{>} 0, $%
\Sigma \left( j,j\right) $ \TEXTsymbol{>} 0, $i<j$; (ii) $\Sigma
\left( i,i\right) $ \TEXTsymbol{>} 0, $\Sigma \left( j,j\right) $
= 0, $i<j$; (iii) $\Sigma \left( i,i\right) $ = $\Sigma \left(
j,j\right) $ = 0, $i<j$. From our previous proof, Fact 2 holds for
case (i). For case (ii), Fact 2 can be guaranteed by letting
$\delta $ be sufficiently small. For case (iii), Fact 2 is
guaranteed from (\ref{Eq Delta Case2}). Thus, Fact 2 is guaranteed
for a sufficiently small $\delta $ and the chosen $\mathbf{\Delta
}$.

Now, consider the following optimization problem:%
\begin{subequations}
\label{xxx}
\begin{equation}
\min_{\mathbf{K}}\text{Tr}\left( \mathbf{G}\left( \gamma \right) \mathbf{K}%
\overline{\mathbf{\Sigma }}\mathbf{K}^{T}\right)   \label{xxxa}
\end{equation}%
s.t.%
\begin{equation}
\left[ \mathbf{K}^{-1}\left( \mathbf{K}^{-1}\right) ^{T}\right] _{\text{diag}%
}=\mathbf{I}\text{.}  \label{xxxb}
\end{equation}%
\end{subequations}
Noting that $\overline{\mathbf{\Sigma }}$ is of full rank and
Facts 2 and 3 hold for $\overline{\mathbf{\Sigma }}$, we see that
the optimal solution to
the above problem is $\mathbf{K}=\mathbf{K}_{opt}^{(\gamma )}$. Now, let $%
\delta \rightarrow 0$. From the continuity of the problem in
(\ref{xxx}), the optimal $\mathbf{K}$ is still given by
$\mathbf{K}_{opt}^{(\gamma )}$. This completes the proof of Lemma
\ref{MarkerLemmaApproxII}.

\newpage
\begin{figure}[tbp]
\centering\includegraphics[width=5in]{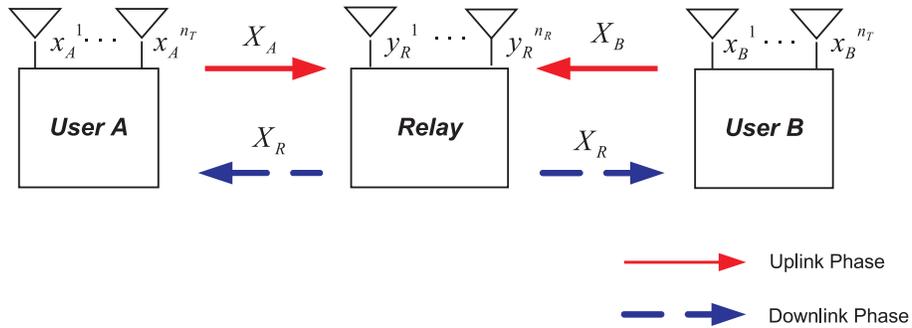}
\caption{Configuration of a MIMO TWRC.}
\label{Fig_Configuration_MIMOTWRC}
\end{figure}

\begin{figure}[tbp]
\centering\includegraphics[width=4.5in]{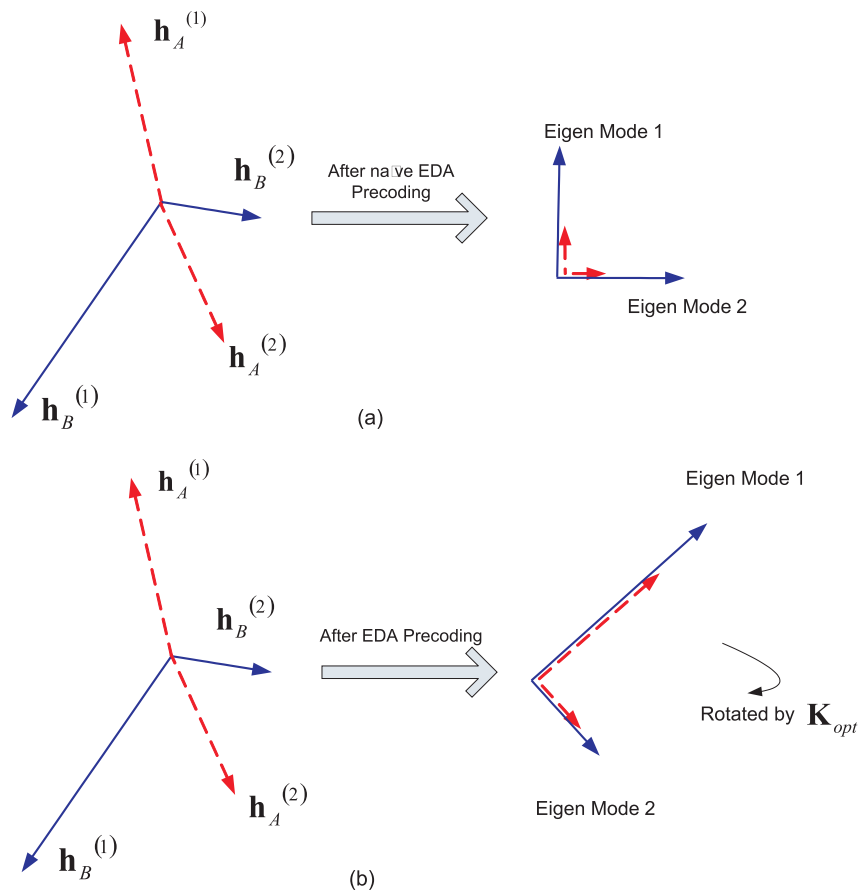}
\caption{Geometrical illustration of (a) Naive EDA precoding and (b) EDA
precoding (with $\mathbf{KK}^{T}=\mathbf{I}$), for a two dimension case.
Here, $\mathbf{H}_{m,R}=[\mathbf{h}_{m}^{\left( 1\right) },\mathbf{h}%
_{m}^{\left( 2\right) }],m\in \left\{ A,B\right\} $. The dashed arrow and
solid arrow denote users $A$ and $B$, respectively. For $A$, since the
correlation of its channel vectors is large, a significant power loss is
suffered in the naive EDA precoding. The proposed EDA precoding can
effectively avoid this loss by introducing a rotation. }
\label{Fig_ZF_ECA}
\end{figure}

%\begin{figure}[tbp]
%\centering\includegraphics[width=6in]{Fig_EDACPNC_Encoder.eps}
%\caption{Encoder architecture of the proposed EDA-PNC scheme for the uplink
%phase. For $m\in \left\{ A,B\right\} $, $\mathbf{W}_{m,i}$ is a codeword in $%
%\mathcal{C}_{m,i}$, $\mathbf{u}_{m,i}$ is a random dither vector known the
%all nodes, $\mathbf{c}_{m,i},i=1,...,n_{R},$are the uncorrelated signal
%streams, and $\mathbf{x}_{m,i}=\left[ x_{m,i}\left( 1\right)
%,...,x_{m,i}\left( n\right) \right] $, $,t=1,...,n$, $i=1,...,n_{R},$ denote
%the signal sequences after the EDA precoding.}
%\label{Fig_EDACPNC_Encoder}
%\end{figure}
%
%\begin{figure}[tbp]
%\centering\includegraphics[width=5in]{Fig_EDACPNC_RelayDecoder.eps}
%\caption{Illustration of relay's operation of the proposed EDA-PNC scheme.
%The relay first estimate the bin-indices $\overline{\mathbf{T}}\mathbf{=[}%
%\overline{\mathbf{T}}_{1},\overline{\mathbf{T}}_{2},...,\overline{\mathbf{T}}%
%_{n_{T}}\mathbf{]}$, which is then encoded and broadcast to the two users.}
%\label{Fig_EDACPNC_RelayDecoder}
%\end{figure}
%
%\begin{figure}[tbp]
%\centering\includegraphics[width=6.3in]{Fig_EDACPNC_DestinationDecoder.eps}
%\caption{Illustration of destinations' decoding operation of the proposed
%EDA-PNC scheme. }
%\label{Fig_EDACPNC_DestinationDecoder}
%\end{figure}

\begin{figure}[tbp]
\centering\includegraphics[width=6in]{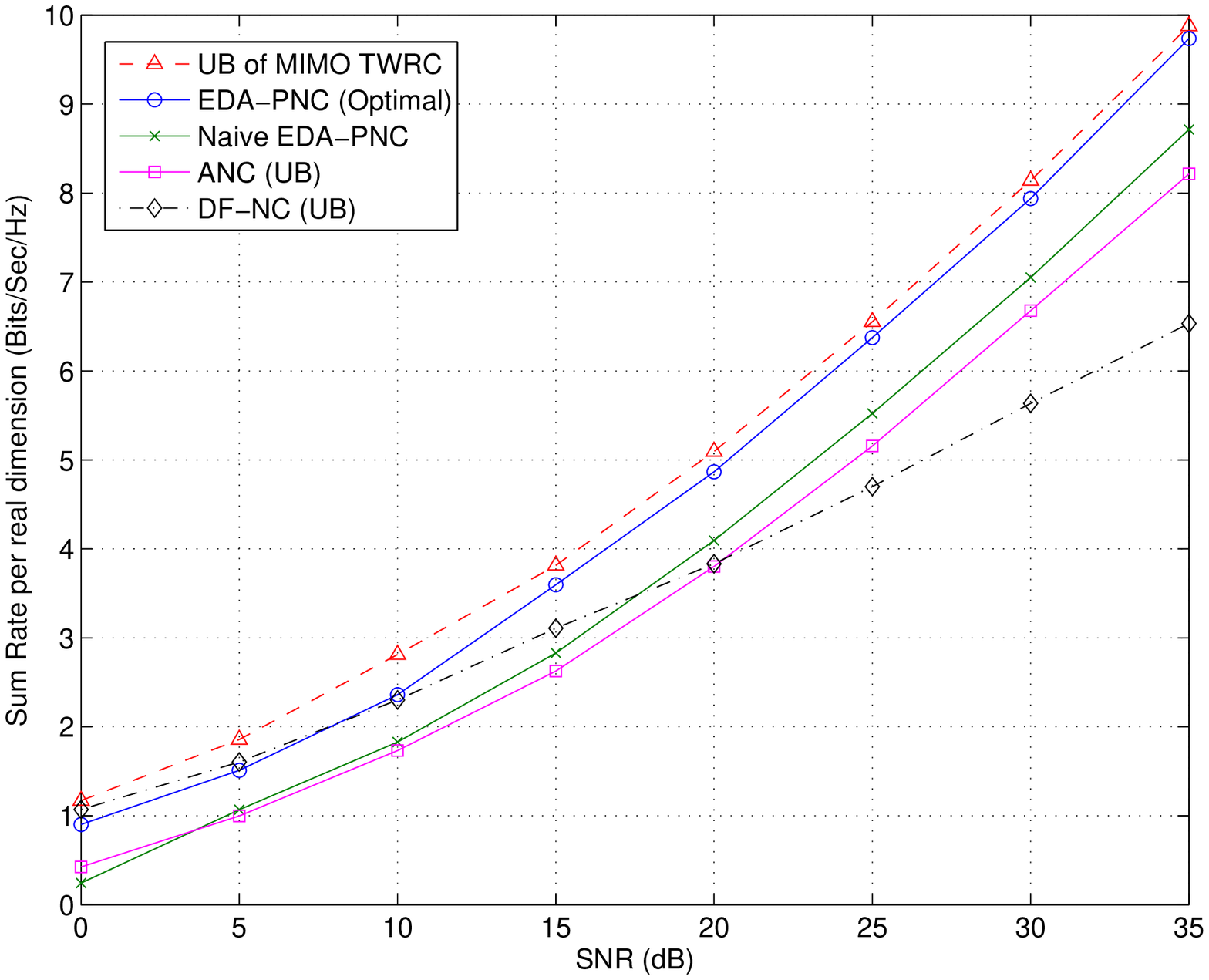}
\caption{Achievable sum-rate of the proposed EDA-PNC scheme for a MIMO TWRC
with $n_{T}=n_{R}=2.$}
\label{Fig_SumRate_2by2}
\end{figure}

\begin{figure}[tbp]
\centering\includegraphics[width=6in]{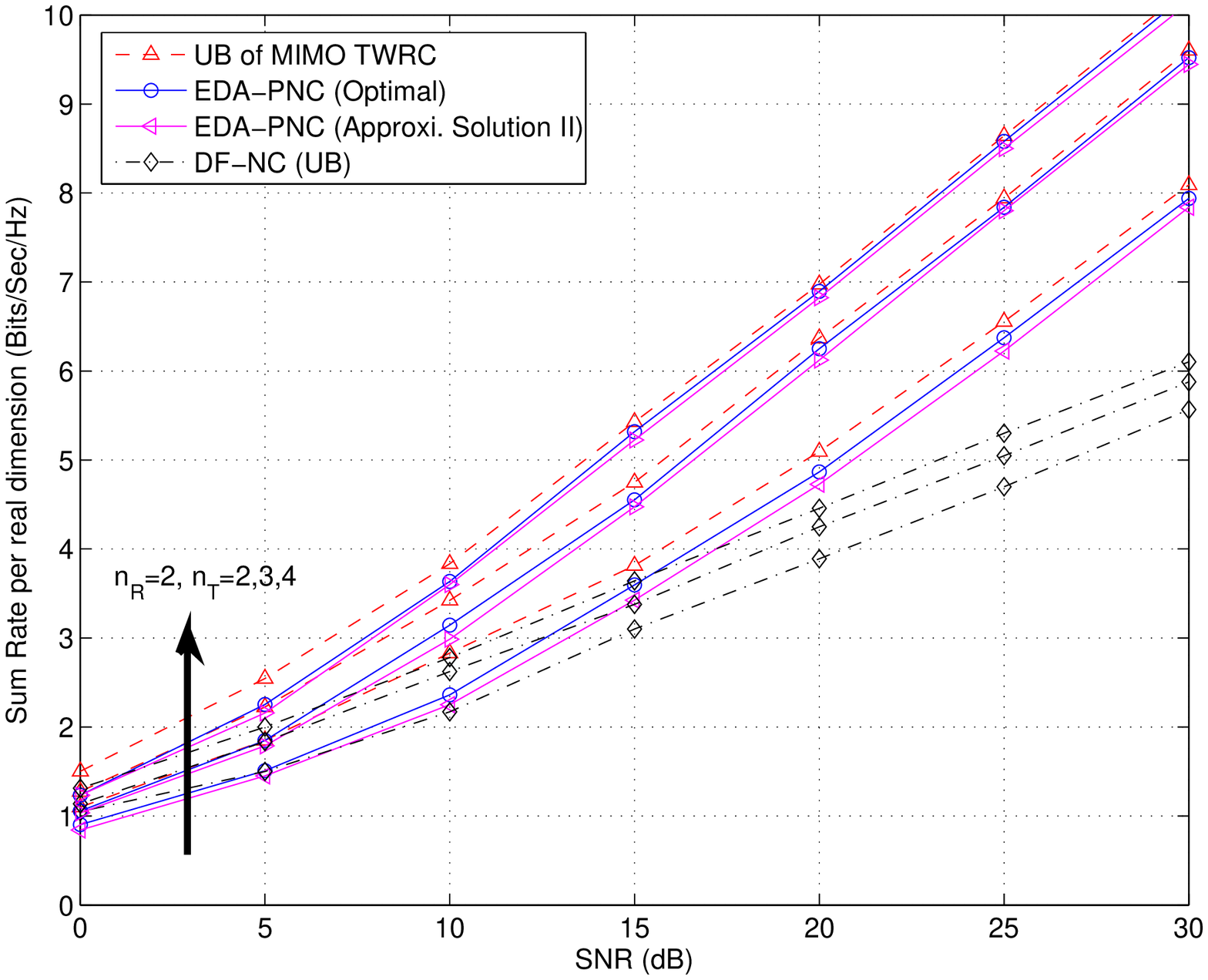}
\caption{Achievable sum-rate of the proposed EDA-PNC scheme for MIMO TWRCs
with $n_{R}=2$, $n_{T}=2,3,4$.}
\label{Fig_SumRate_nR2_Various_nT}
\end{figure}

\begin{figure}[tbp]
\centering\includegraphics[width=6in]{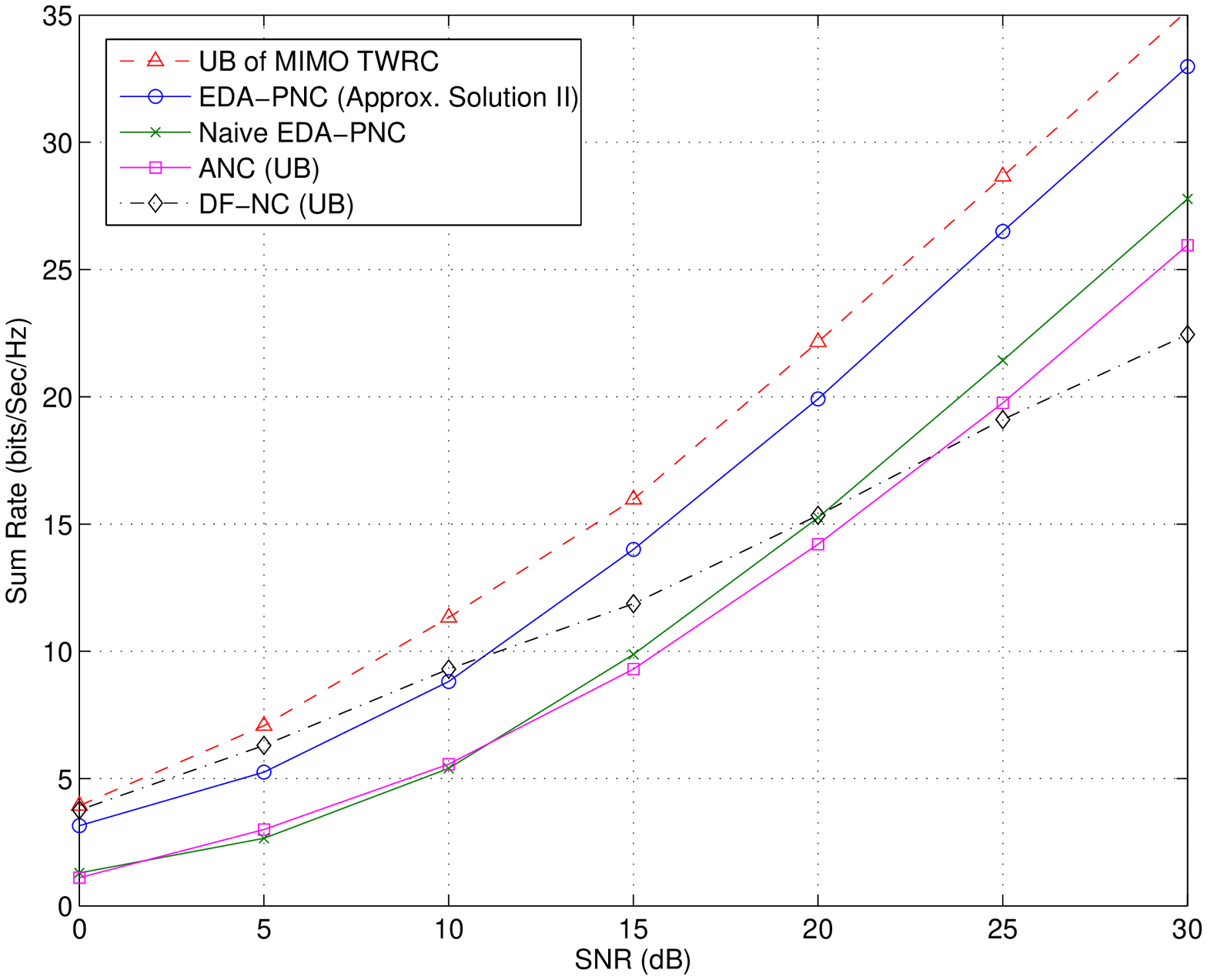}
\caption{Achievable sum-rate of the proposed EDA-PNC scheme for a MIMO TWRC
with $n_{T}=n_{R}=4.$}
\label{Fig_AchievableRate_MIMO_Nr4_VariousSchemes}
\end{figure}
\begin{figure}[tbp]
\centering\includegraphics[width=6in]{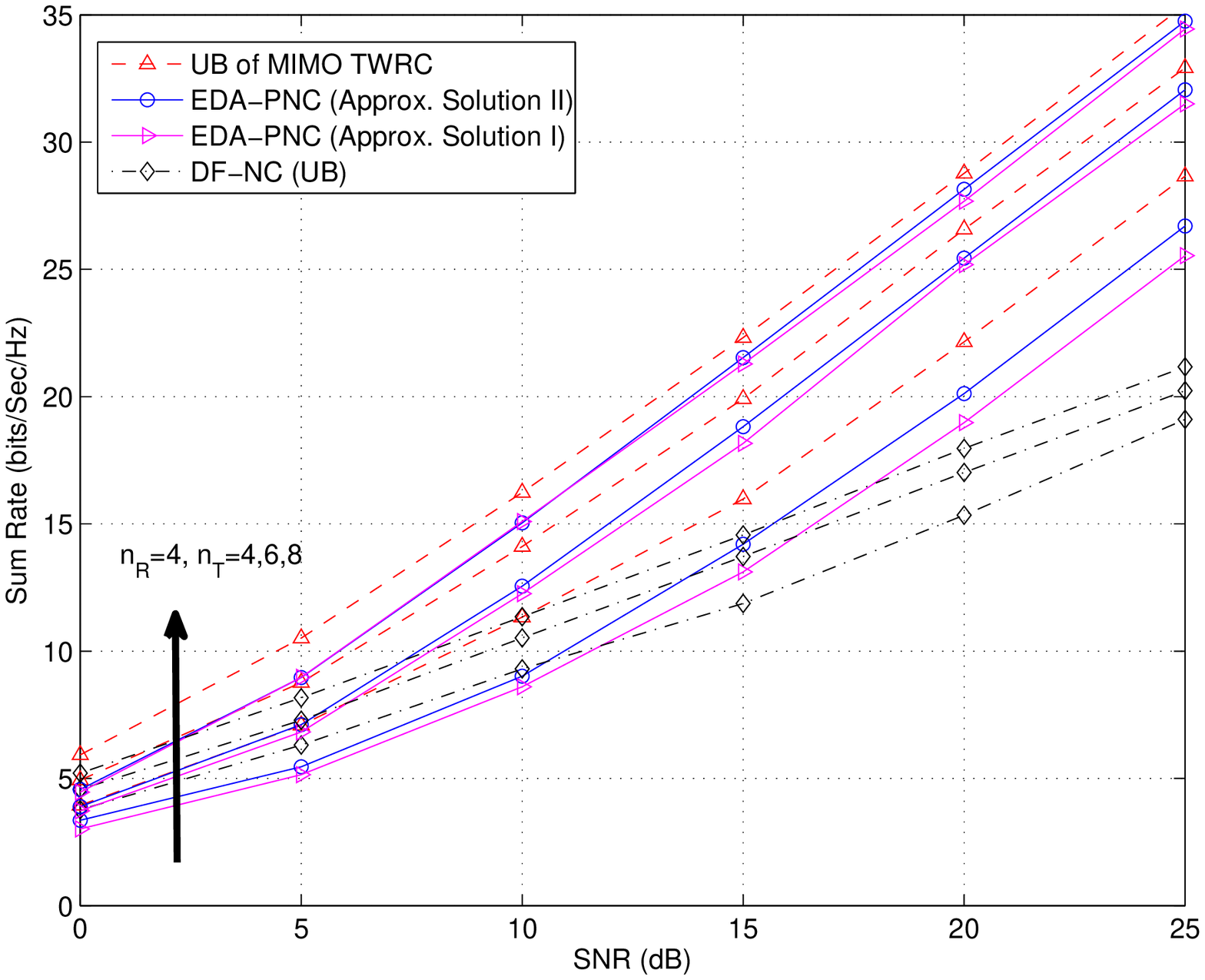}
\caption{Achievable sum-rate of the proposed EDA-PNC scheme for MIMO TWRCs
with $n_{R}=4$, $n_{T}=4,6,8$.}
\label{Fig_AchievableRate_MIMO_Nr4}
\end{figure}

\begin{figure}[tbp]
\centering\includegraphics[width=6in]{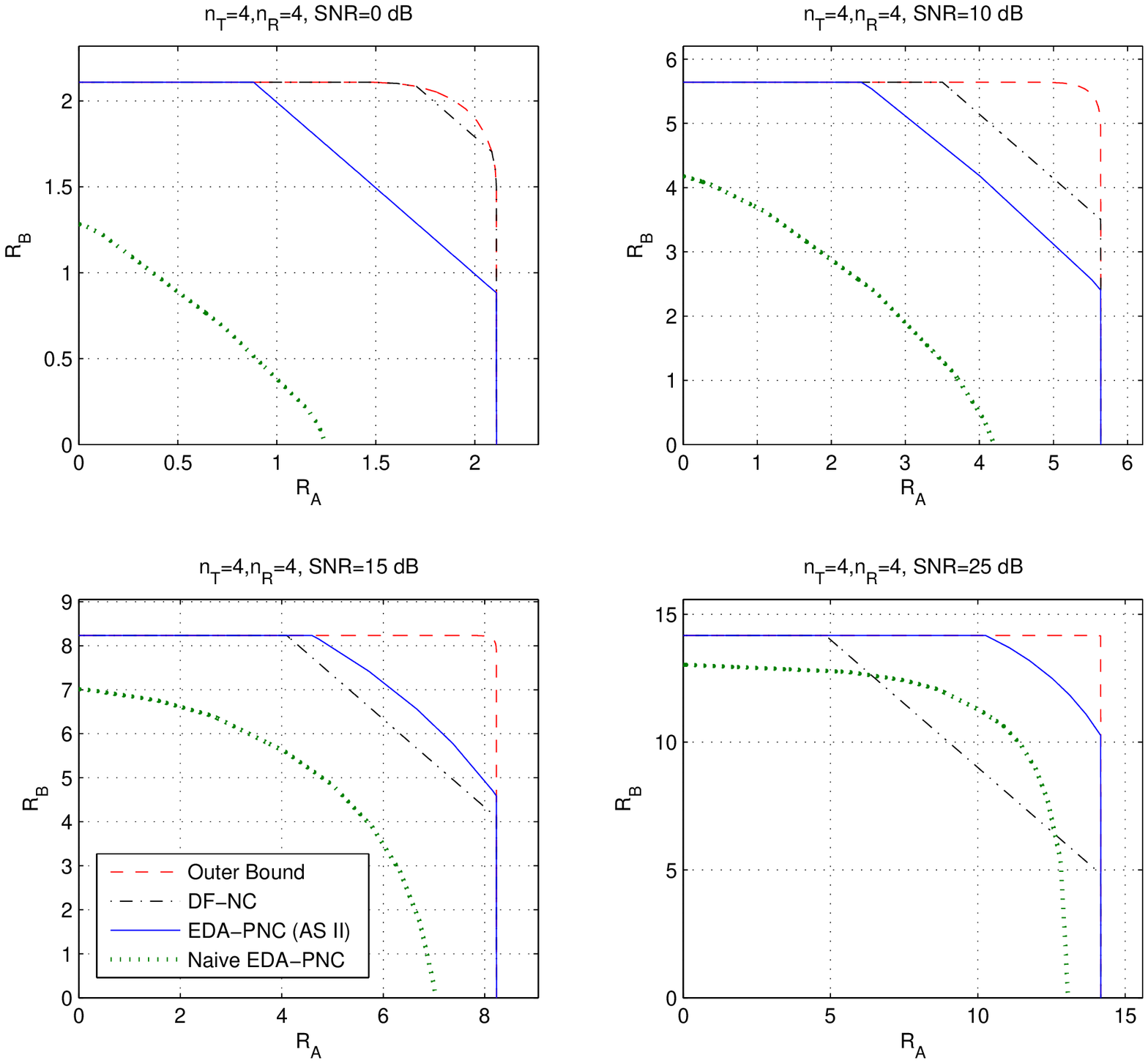}
\caption{Achievable rate-region of the proposed EDA-PNC scheme for a MIMO
TWRC with $n_{T}=n_{R}=4,$ where $SNR=0,10,15,25$ dB.}
\label{Fig_RateRegion_nT4nR4}
\end{figure}

\begin{figure}[tbp]
\centering\includegraphics[width=6in]{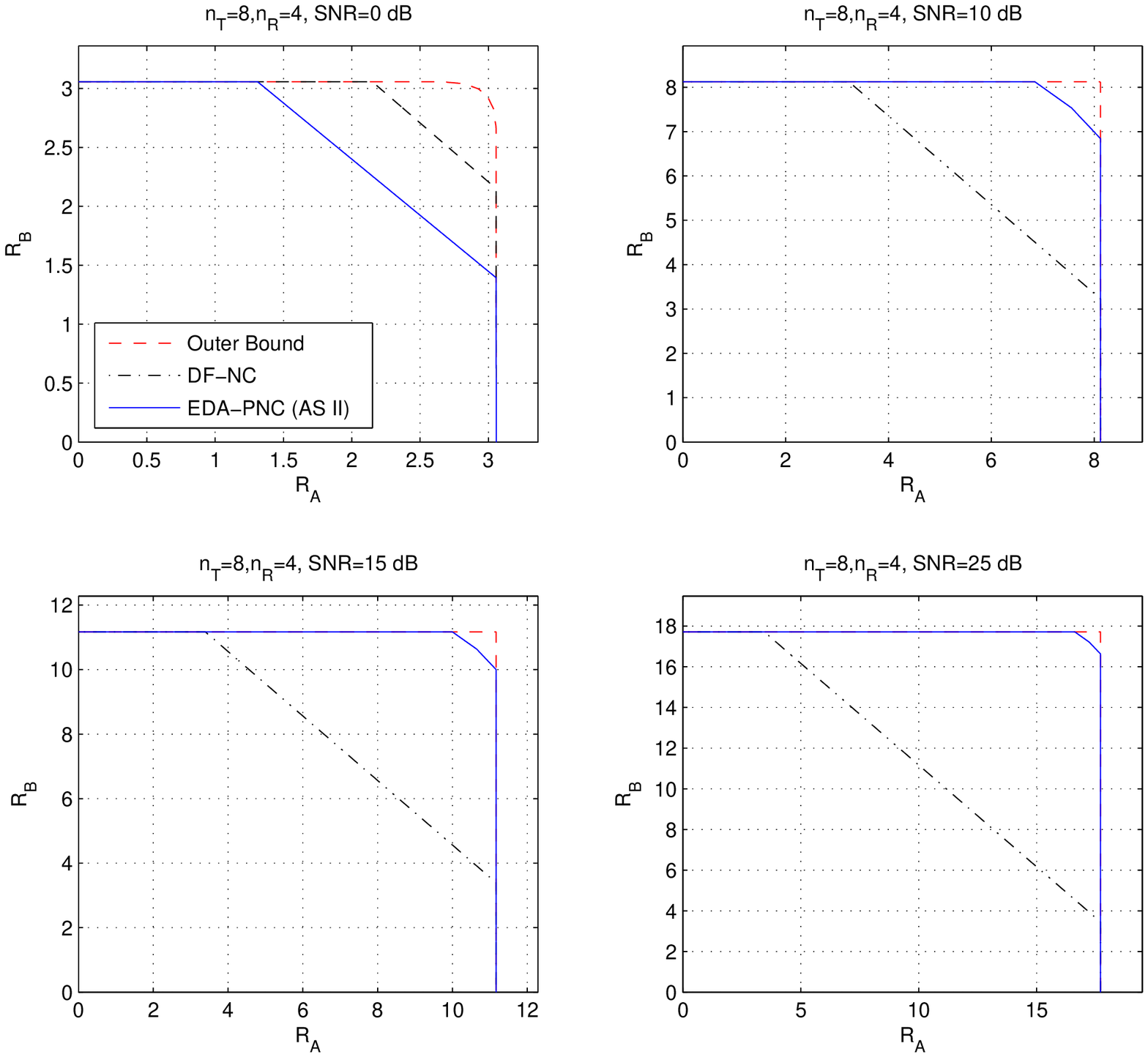}
\caption{Achievable rate-region of the proposed EDA-PNC scheme for a MIMO
TWRC with $n_{T}=8$, $n_{R}=4,$ where $SNR=0,10,15,25$ dB.}
\label{Fig_RateRegion_nT8nR4}
\end{figure}


\begin{thebibliography}{99}
\bibitem{ZhangMobicom06} S. Zhang, S. Liew, and P. Lam, \textquotedblleft
Physical-layer network coding,\textquotedblright\ in ACM Mobicom `06.

\bibitem{NamIT09} W. Nam, S. Chung, Y. H. Lee, \textquotedblleft Capacity of
the Gaussian two-way relay channel to within 1/2 bit\textquotedblright ,
\textit{IEEE Trans. Inform.\ Theory}, vol. 56, no. 11, pp. 5488-5494, Nov.
2010.

\bibitem{WilsonIT10} M. P. Wilson, K. Narayanan, H. D. Pfister and A.
Sprintson, \textquotedblleft Joint physical layer coding and network coding
for bidirectional relaying\textquotedblright , \textit{IEEE Trans. Inform.
Theory}, vol. 56, no. 11, pp. 5641-5654, Nov. 2010.

\bibitem{KattiSigcomm07} S. Katti, S. Gollakota, and D. Katabi,
\textquotedblleft Embracing Wireless Interference: Analog Network
Coding\textquotedblright , ACM SIGCOMM'07.

\bibitem{RuiZhangJSAC09} R. Zhang, Y.-C. Liang, C. C. Chai, and S. Cui,
\textquotedblleft Optimal beamforming for two-way multi-antenna relay
channel with analogue network coding,\textquotedblright\ \textit{IEEE
Journal Select. Area. Comm.}, Vol. 27, No. 5, pp. 699-712, June 2009.

\bibitem{XuICCASP2010} S. Xu and Y. Hua, \textquotedblleft Source-relay
optimization for a two-way MIMO relay system\textquotedblright , \textit{%
Proc. IEEE ICCASP 2010}, pp. 3038-3041.

\bibitem{GunduzAsilomar2008} D. Gunduz, A. Goldsmith and H. V. Poor,
\textquotedblleft MIMO two-way relay channel: diversity-multiplexing
tradeoff analysis\textquotedblright , Asilomar 2008.

\bibitem{Foschini96} G. Foschini, \textquotedblleft Layered space-time
architecture for wireless communication in a fading environment when using
multi-element antennas,\textquotedblright\ \emph{Bell Labs Technical Journal}%
, Autumn 1996, pp. 41-59.

\bibitem{AhlswedeIT2000} R. Ahlswede, N. Cai, S.-Y. R. Li, and R. W. Yeung,
\textquotedblleft Network information flow,\textquotedblright\ \textit{IEEE
Trans. Inform. Theory}, vol. 46, pp. 1204--1216, Oct. 2000.

\bibitem{LanemanIT03} J. N. Laneman, D. N. C. Tse, and G. W. Wornell,
\textquotedblleft Cooperative diversity in wireless networks: Efficient
protocols and outage behavior,\textquotedblright\ \textit{IEEE Trans. Inf.
Theory}, vol. 50, no. 12, pp. 3062--3080, Dec. 2004.

\bibitem{CoverTextBook} T. M. Cover and J. A. Thomas, \textit{Elements of
Information Theory}, New York: Wiley, 1991.

\bibitem{TSETextBook} {D. Tse and P. Visanath, \textit{Fundamentals of
wireless communicationts}, Cambridge University Press, 2006. }

\bibitem{ZamirIT02} R. Zamir, S. Shamai and U. Erez, \textquotedblleft
Nested linear/lattice codes for structured multiterminal
binning\textquotedblright , \textit{IEEE Trans. Inform. Theory}, vol. 48,
no. 6, pp. 1250-1276, June 2002.

\bibitem{ErezIT04} U. Erez and R. Zamir, \textquotedblleft Achieving 1/2
log(1 + SNR) on the AWGN channel with lattice encoding and
decoding,\textquotedblright\ \textit{IEEE Trans. Inform. Theory}, vol. 50,
pp. 2293-2314, Oct. 2004.

\bibitem{MatrixAnalysis} R. A. Horn and C. R. Johnson, \textit{Matrix
Analysis}. Cambridge University Press, 1985.

\bibitem{MyselfJSAC2011} T. Yang and X. Yuan, \textquotedblleft Multi-stream
physical layer network coding with Eigen-direction alignment precoding for
MIMO two-way relay channels\textquotedblright , in preparation, 2011.

\bibitem{ZhangJSAC09} S. Zhang and S.-C. Liew, \textquotedblleft Channel
coding and decoding in a relay system operated with physical-layer network
coding\textquotedblright , \textit{IEEE Journal Select. Area. Comm.}, vol.
27, no. 5, pp. 788-796, June 2009.

\bibitem{NazerIT07} B. Nazer and M. Gastpar, \textquotedblleft Computation
over multiple-access channels,\textquotedblright\ \textit{IEEE Trans.
Inform. Theory}, vol. 53, pp. 3498--3516, Oct. 2007.

\bibitem{GoldsmithTextBook} {A. Goldsmith, \textquotedblleft Wireless
communicationts\textquotedblright , \textit{Cambridge University Press},
2005.}

\bibitem{YuIT04} W. Yu, W. Rhee, S. Boyd and J. M. Cioffi, \textquotedblleft
Iterative water-filling for Gaussian vector multiple-access
channels\textquotedblright , \textit{IEEE Trans. Inform. Theory}, vol. 50,
pp. 145--152, Jan. 2004.

\bibitem{TangTWC07} X. Tang and Y. Hua, \textquotedblleft Optimal design of
non-regenerative MIMO wireless relays\textquotedblright , \textit{IEEE
Trans. Wirele. Comm.}, vol. 6, pp. 1398--1407, Apr. 2007.

\bibitem{NamIT09 P2} W. Nam, S.-Y. Chung and Y. H. Lee, \textquotedblleft
Nested latice codes for Gausian relay networks with
interference\textquotedblright , submitted to \textit{IEEE Trans. Inform.
Theory. }

\bibitem{ErezIT05} U. Erez, S. Litsyn and R. Zamir, \textquotedblleft
Lattices which are good for (almost) everything\textquotedblright , \textit{%
IEEE Trans. Inform. Theory}, vol. 51, pp. 3401--3416, Oct. 2005.

\bibitem{ForneyAllerton03} G. D. Forney Jr., \ \textquotedblleft On the role
of MMSE estimation in approaching the information theoretic limits of linear
Gaussian channels: Shannon meets wiener,\textquotedblright\ Proc. 41st
Annual Allerton Conference, Oct. 2003.

\bibitem{AbrudanTSP08} T. E. Abrudan, J. Eriksson and V. Koivunen,
\textquotedblleft Steepest descent algorithms for optimization under unitary
matrix constraint\textquotedblright , \textit{IEEE Trans. Sig. Proc.}, vol.
56, no. 3, pp. 1134-1147, Mar. 2008.

\bibitem{CVXSoftware} S. Boyd and L. Vandenberghe, \textquotedblleft Convex
optimization\textquotedblright , \textit{Cambridge University Press}, 2004.

\bibitem{LouieTWC10} R. H. Y. Louie, Y. Li and B. Vucetic, \textquotedblleft
Practical physical layer network coding for two-way relay channels:
performance analysis and comparison\textquotedblright , \textit{IEEE Trans.
Wireless Comm.}, vol. 9, no. 2, pp. 764-777, Feb. 2010.

\bibitem{ZhouTcom10} Q. F. Zhou, Y. Li, F. C. M. Lau and B. Vucetic,
\textquotedblleft Decode-and-forward two-way relaying with network coding
and opportunistic relay selection\textquotedblright , \textit{IEEE Trans.
Commun.}, vol. 58, no. 11, pp. 3070-3076, Nov. 2010.

\bibitem{AkinoJSAC09} T. Koike-Akino, P. Popovski and V. Tarokh,
\textquotedblleft Optimized constellations for two--way wireless relaying
with physical network coding\textquotedblright , \textit{IEEE Jour. Select
Area. Commun..}, vol. 27, no. 5, pp. 773-787, June 2009.

\bibitem{Rohatgi} V. K. Rohatgi and A. K. Md. E. Saleh, \textquotedblleft An
introduction to probability and statistics\textquotedblright , \textit{%
Wiley-Interscience}, second edition, 2000.

\bibitem{KayTextbook} Steven M. Kay, \textquotedblleft Fundamentals of
Statistical Signal Processing\textquotedblright , \textit{Prentice-Hall PTR}%
, 1993.

\bibitem{HornTextbook} R. A. Horn and C. R. Johnson, \textquotedblleft
Matrix Analysis\textquotedblright , \textit{Cambridge Unversity Press}, 1990.

\bibitem{YuanANC} X. Yuan and T. Yang, \textquotedblleft Near optimal
precoding for a MIMO two-way relay system operated with analog network
coding\textquotedblright , in preparation.
\end{thebibliography}
\end{document}